\documentclass [journal,onecolumn,11pt]{IEEEtran}
\usepackage{amsfonts,amsmath,amssymb}
\usepackage{indentfirst, setspace}
\usepackage{url,float}
\usepackage{color}
\usepackage[a4paper,ignoreall]{geometry}
\usepackage{longtable}
\usepackage{graphicx}
\usepackage[a4paper,ignoreall]{geometry}
\usepackage{multicol}
\usepackage{stfloats}
\usepackage{enumerate}
\usepackage{cite}
\usepackage{amsthm}
\usepackage{mathrsfs}
\usepackage{multirow}
\usepackage{amssymb}
\usepackage{makecell}
\usepackage[square, comma, sort&compress, numbers]{natbib}

\def\qu#1 {\fbox {\footnote {\ }}\ \footnotetext { From Qu: {\color{red}#1}}}
\def\hqu#1 {}
\def\kq#1 {\fbox {\footnote {\ }}\ \footnotetext { From KangQuan: {\color{blue}#1}}}
\def\hkq#1 {}

\newtheorem{Th}{Theorem}[section]
\newtheorem{Cor}[Th]{Corollary}
\newtheorem{Prop}[Th]{Proposition}

\newtheorem{Lemma}[Th]{Lemma}
\newtheorem{Def}[Th]{Definition}
\newtheorem{example}[Th]{Example}

\newtheorem{Rem}[Th]{Remark}

\newcommand{\tr}{{\rm Tr}}

\newcommand{\gf}{{\mathbb F}}
\newcommand{\C}{{\mathcal C}}

\newcommand{\bu}{{\bf u}}
\newcommand{\bv}{{\bf v}}

\newcommand{\wt}{{\rm wt}}
\newcommand{\V}{{\mathcal{V}}}
\newcommand{\A}{{\mathcal{A}}}
\newcommand{\supp}{{\rm supp}}

\makeatletter
\newcommand{\figcaption}{\def\@captype{figure}\caption}
\newcommand{\tabcaption}{\def\@captype{table}\caption}
\makeatother

\begin{document}
	\title{Constructions of binary self-orthogonal singly-even minimal linear codes violating the Aschikhmin-Barg condition with few weights}
	\author{
    {Kangquan Li, Hao Chen, Wengang Jin and Longjiang Qu}\\
	\thanks{ Kangquan Li, Wengang Jin, and Longjiang Qu are with the College of Science,
		National University of Defense Technology, Changsha, 410073, China (e-mail: likangquan11@nudt.edu.cn, jinwengang0110@126.com, ljqu\_happy@hotmail.com). Hao Chen is with the College of Information Science and Technology, Jinan University, Guangzhou, 510632, China (e-mail: haochen@jnu.edu.cn).
		This work is supported by the National Key R\&D Program of China (No. 2024YFA1013000) and the National Natural Science Foundation of China  (Nos. 12571579, 12525115 and 62032009).
		{\emph{(Corresponding author: Longjiang Qu)}}
        }
	}
	\maketitle{}

\begin{abstract}
 Recent research has extensively explored the construction of binary self-orthogonal (SO) linear codes, minimal linear codes violating the Ashikhmin–Barg (AB) condition, and linear codes with few weights, owing to their broad applications. However, few constructions simultaneously satisfy all these properties. This paper addresses this gap by constructing binary linear codes that are self-orthogonal, singly-even, minimal, violate the AB condition, and have few weights.

We first establish a simple yet powerful necessary and sufficient condition for a binary linear code to be SO, leading to a complete characterization of singly-even codes in this family. We further derive necessary and sufficient conditions on Boolean and vectorial Boolean functions for generating such codes via a standard construction method. Building on this foundation, we propose three general frameworks for constructing binary SO singly-even minimal non-AB linear codes with few weights. The first two approaches are based on designing Boolean and vectorial Boolean functions that simultaneously satisfy multiple conditions. The third method generates new SO codes from existing ones. As a result, we obtain many infinite classes of binary self-orthogonal singly-even minimal linear codes violating the AB condition with few weights and fully determined weight distributions. Particularly, numerical results show that some duals of our codes are optimal or near-optimal.

\end{abstract}


\section{Introduction}
Let $\gf_2^n$ be the linear space of $\gf_2$ with dimension $n$. A non-empty subset $\C$  of $\gf_2^n$ is called a binary $[n,k]$-linear code if $\C$ is a subspace with dimension $k$. Any element of $\C$ is called a codeword. For any codeword $\bu=(u_1,u_2,\ldots,u_n)\in\C$, its (Hamming) weight $\wt(\bu)$ is the cardinality of its support, i.e., $\wt(\bu)=\#\supp(\bu)$ with $\supp(\bu) = \{ i: 1 \le i\le n ~|~ u_i =1 \}$. For any binary $[n,k]$-linear code $\C$, the minimum (Hamming) distance of $\C$ is $d=\min_{\bu\neq0}\{\wt(\bu): \bu\in\C\}$. If the minimum distance $d$ of $\C$ is known, then $\C$ can be called a binary $[n,k,d]$-linear code.  { It is sometimes said to be optimal (resp. near-optimal) when  $d$ (resp. $d+1$) achieves the maximum possible value for given parameters $n$ and $k$. }
Moreover, let $A_i$ denote the number of codewords with weight $i$ in $\C$. The weight enumerator of $\C$ is defined as $1+A_1z+A_2z^2+\cdots+A_nz^n$. The sequence $(1, A_1, A_2,\ldots, A_n)$ is called the weight distribution of $\C$. Moreover, $\C$ is said to be a $t$-weight code or a code with $t$ weights if the number of nonzero $A_i$ in the sequence $(1, A_1, A_2,\ldots, A_n)$ is equal to $t$.

Binary linear codes are of significant interest due to their applications in diverse fields such as data communication, cryptography, and combinatorics. Within this broad class, certain specialized types of codes have been extensively studied for specific purposes. Among these, self-orthogonal (SO) codes are particularly important. Their significance arises notably in quantum information theory, where they facilitate the construction of quantum error-correcting codes \cite{calderbank1997quantum,calderbank1998quantum}. Furthermore, SO codes find applications in constructing even lattices \cite{wan1998characteristic}, linear complementary dual (LCD) codes \cite{massey1998orthogonal}, and other combinatorial structures. Formally, let $\C$ be a binary linear code. Its dual code $\C^{\perp}$ is defined as:
 $$\C^{\perp} = \left\{  \bv: \bv \in \C ~|~ \bu \cdot \bv = 0 ~\text{for all}~ \bu\in \C \right\}.$$
$\C$ is called  {self-orthogonal} (SO) if $\C \subseteq \C^{\perp}$.  Given these applications, the construction of binary SO linear codes has been an active research area in recent years, e.g. \cite{zhou2018binary,wu2020binary,du2020constructions,wang2024self,shi2025binary,xie2025constructions}.

Although several construction methods for binary SO linear codes exist, their most fundamental property is characterized by the following theorem:
\begin{Th}
    \cite{huffman2010fundamentals}
    \label{basic_property}
    Let $\C$ be a binary linear code. If $\C$ is SO, then for any $\bu\in\C$, $\wt(\bu)$ is even. If for any $\bu\in\C$, $\wt(\bu)\equiv 0\pmod 4$, then $\C$ is SO.
\end{Th}
Theorem \ref{basic_property} establishes a direct relationship between the self-orthogonality of $\C$ and the divisibility of its weights. For clarity, we adopt standard terminology.
\begin{itemize}
    \item A binary code is {doubly-even} if all codewords have weights divisible by $4$.
    \item A binary code is {singly-even} if all codewords have even weights, but not all are divisible by $4$.
\end{itemize}
From Theorem \ref{basic_property}, we know that if a binary linear code $\C$ is doubly-even, then $\C$ is SO; and if $\C$ is SO, then the weights of codewords in $\C$ must be even. Therefore, a natural question is {how to judge whether a binary singly-even linear code is self-orthogonal or not}.

Beyond self-orthogonality, following the seminal work in \cite{ding2018minimal}, significant attentions, e.g. \cite{mesnager2019several,mesnager2023several,zhang2019four,li2020four,li2022minimal}, have focused on constructing  {minimal linear codes that violate the Aschikhmin-Barg (AB) condition} due to their applications in secret sharing  \cite{massey1993minimal} and secure two-party computation \cite{chabanne2013towards}. In addition, binary linear codes with  {few weights} have attracted interest due to their applications across multiple domains, including secret sharing schemes \cite{carlet2005linear,yuan2005secret}, authentication codes \cite{ding2005coding}, association schemes \cite{calderbank1984three}, and strongly regular graphs \cite{calderbank1986geometry}. Numerous such codes have been constructed through diverse approaches, see e.g. \cite{ding2016construction,ding2015linear,wang2021some,li2021binary,mesnager2021linear}.
{However, to the best of our knowledge, there is only one explicit construction of binary linear codes that, at the same time, satisfy that they are self-orthogonal, singly-even, minimal, violate the AB condition, and have few weights, see Theorem \cite[Theorem 3]{jin2025several}.}

{
{Motivated by the aforementioned considerations, this paper considers the constructions of infinite families of binary linear codes that are simultaneously self-orthogonal (SO), singly-even, minimal, non-AB (violating the Ashikhmin–Barg condition), and have few weights. The main contributions of this work are summarized as follows.}
\begin{enumerate}[(1)]
    \item 
    We first establish a simple yet powerful necessary and sufficient condition (Lemma \ref{main_th}) for a binary linear code to be SO. Based on this, a complete characterization of binary SO singly-even linear codes is provided in Proposition \ref{simple_lemma}. Furthermore, necessary and sufficient conditions are given for Boolean functions $f$ such that the codes $\C_f$ defined in \eqref{first_method} are SO and singly-even (Lemma \ref{Cf_SO_SE}), and for vectorial Boolean functions $F$ such that $\C_F$ is SO (Lemma \ref{CF_SO}).
    \item   Using Lemma \ref{Cf_SO_SE}, we derive a necessary and sufficient condition for a known class of minimal codes violating the AB condition \cite{ding2018minimal} to be SO and singly-even. Moreover, a general method is presented for constructing Boolean functions satisfying the condition in Lemma \ref{Cf_SO_SE}. 
Additionally, several infinite families of binary SO singly-even minimal linear codes violating the AB condition with few weights are explicitly constructed. The complete weight distributions of these codes are determined.
    \item  Building upon Lemma \ref{CF_SO}, we investigate a generic construction of binary linear codes $\C_F$ as defined in \eqref{first_method}, based on vectorial Boolean functions of the form $F(x) = (f(x), G(x))$, where $f$ satisfies conditions ($\mathbf{c}_1$) and ($\mathbf{c}_2$). We first provide sufficient conditions for a product $f = f_1 f_2$ of two bent functions $f_1$ and $f_2$ to satisfy ($\mathbf{c}_1$) and ($\mathbf{c}_2$) (Proposition \ref{bvg-t}). Then, by selecting $f_1$ and $f_2$ from the well-known Maiorana–McFarland class, we explicitly construct several infinite families of binary SO singly-even minimal linear codes that violate the AB condition and have few weights. The complete weight distributions of these codes are determined.
        \item   Finally, we generalize the construction approach introduced in \cite{chen2025minimal} to generate diverse binary SO linear codes from pairs of existing SO codes. Through this general method, we obtain an infinite family of binary SO singly-even two-weight codes, along with two infinite families of binary SO singly-even four-weight codes. All the resulting codes are minimal and violate the AB condition, and their weight distributions are fully determined.
\end{enumerate}}

{The remainder of this paper is organized as follows. Section \ref{Preliminaries} establishes foundational concepts for vectorial Boolean functions and binary linear codes. On this basis, in Section \ref{Characterization}, we present our key characterization of binary SO linear codes and derive the necessary and sufficient condition for binary SO singly-even linear codes. Necessary and sufficient conditions are given for Boolean functions $f$ such that the codes $\C_f$ defined in \eqref{first_method} are SO and singly-even, and for vectorial Boolean functions $F$ such that $\C_F$ is SO.
Sections \ref{Constructions_from_Boolean} and \ref{Constricton_from_vectorial_Boolean} construct some infinite families of binary SO singly-even minimal linear codes violating the AB condition with few weights from Boolean functions and vectorial Boolean functions, respectively. 
 Section \ref{first_general_approach} develops a general construction method for binary SO codes by combining existing SO codes. This approach generates three infinite classes of binary singly-even SO codes that simultaneously satisfy: minimality, violation of the AB condition, and few-weight properties.
Finally, concluding remarks and research directions appear in Section \ref{conclusion}.}


\section{Preliminaries}
\label{Preliminaries}

We first unify some frequent notions as follows.
\begin{itemize}
\item $\gf_2^n$: the linear space of $\gf_2$ with dimension $n$. Note that $\gf_2^n$ can be endowed with the structure of the field $\gf_{2^n}$;
\item $\gf_{2^m}^{*}$: all nonzero elements in the finite field $\gf_{2^m}$;
\item $f^*$: the dual of $f$, where $f$ is bent;
    \item $\C$: a binary $[n,k,d]$-linear code;
    \item for a given set $E$, $\#E$ denotes its cardinality;
    \item for two sets $E,S$, $E\backslash S$ denotes all elements in $E$ and not in $S$;
    \item for $i\in\{0,1\}$, $\mathbf{i}_n = (i,i,\ldots,i)\in\gf_2^n$. We also write $\mathbf{i}$ for convenience;
    \item for any $\bu = (u_1,u_2,\ldots,u_n)\in\gf_2^n$, $\supp(\bu) = \{ i: 1 \le i\le n ~|~ u_i =1 \}$;
    \item for any $\bu = (u_1,u_2,\ldots,u_n)\in\gf_2^n$, $\wt(\bu)=\#\supp(\bu)$;
    \item for any $\bu = (u_1,u_2,\ldots,u_n), \bv = (v_1, v_2,\ldots, v_n) \in\gf_2^n$, $\bu\cdot \bv = \sum_{i=1}^nu_iv_i$ denotes the inner product of $\bu$ and $\bv$. If $\gf_2^n$ is identified to the finite field $\gf_{2^n}$, then $\bu\cdot \bv = \tr_{2^n}(\bu\bv)$, where $\tr_{2^n}(\cdot)$ is the trace function from $\gf_{2^n}$ to $\gf_2$.
\end{itemize}

\subsection{Basic knowledge of (vectorial) Boolean functions}

We first recall that a function $F$ from $\gf_{2^m}$ (or $\gf_2^m$)  to $\gf_{2^k}$  (or $\gf_2^k$)   is called a $(m,k)$-function or vectorial Boolean function. Moreover, when $k=1$, a $(m,k)$-function is called Boolean and in this case, we always use the notion $``f"$. For comprehensive results about (vectorial) Boolean functions, interested readers refer to \cite{carlet2021boolean}.

Let $F$ be a function from $\gf_{2^m}$ to  $\gf_{2^k}$. Then the Walsh transform of $F$ is defined by
$$W_F(a, b) = \sum_{x\in\gf_{2^m}} (-1)^{\tr_{2^k}(a F(x)) + \tr_{2^m}(b x)}$$
for all $(a, b)\in\gf_{2^k}^{*}\times \gf_{2^m}$, where $\tr_{2^m}(\cdot)$ is the trace function from $\gf_{2^m}$ to $\gf_2$.

If $F$ is defined from $\gf_2^m$ to $\gf_2^k$, then its Walsh transform is
$$W_F(a, b) = \sum_{x\in\gf_2^m} (-1)^{a \cdot F(x)+b\cdot x},$$
where $b\cdot x$ denotes the inner product of $b$ and $x$.

If $f$ is a Boolean function from $\gf_2^m$ to $\gf_2$, its Walsh transform is
$$W_f(b) = \sum_{x\in\gf_2^m} (-1)^{f(x)+b\cdot x}.$$ By the well-known Parseval's relation $\sum_{b\in\gf_2^m}W_f(b)^2 = 2^{2m}$, we have $|W_f(b)|\ge 2^{\frac{m}{2}}$ and thus we can introduce the concept of bent functions, which has wide applications in cryptography, coding theory and so on. For comprehensive knowledge about bent functions, readers refer to \cite{mesnager2016bent}.

\begin{Def}
	\cite{rothaus1976bent}
	A Boolean function $f$ from $\gf_2^m$ to $\gf_2$ is called bent if for any $b\in\gf_2^m$, $W_f(b)\in\left\{\pm 2^{\frac{m}{2}}\right\}$.
\end{Def}

It is well-known that bent functions always occur in pairs. In fact, given a bent function $f$ over $\mathbb{F}_2^m$ with $m$ even, we define its dual function, denoted by $f^*$, considering the signs of the values of the Walsh transform of $f$. More precisely, $f^*$ is defined by
\begin{eqnarray}\label{that f}
	W_{f}(x)=2^{\frac{m}{2}}(-1)^{f^*(x)}.
\end{eqnarray}
 Thus for any bent function $f$ from $\gf_2^m$ to $\gf_2$, we have
$$W_f(0) =2^{\frac{m}{2}}(-1)^{f^*(0)}.$$

{\begin{Rem}
		\label{f_f+1}
	For any bent function $f$ from $\gf_2^m$ to $\gf_2$ , let $\bar{f}=f+1$. Then it is trivial that for any $b\in\gf_{2^m}$, $W_{\bar{f}}(b) = -W_f(b)$, implying that $\bar{f}$ is also bent and $\bar{f}^{*}=-f^{*}$.
\end{Rem}}

A well-known class of bent functions is the Maiorana-McFarland (MM for short) class $\mathcal{M}$ \cite{dillon1974elementary,mcfarland1973family}, which gives by far the largest class of bent functions defined by the set of all the Boolean functions on {$\gf_2^m = \left\{ (x,y): x, y \in\gf_2^t \right\}$} with $m=2t$, of the form:
\begin{equation}
	\label{MM}
	f(x, y) = x\cdot \pi(y) + g(y),
\end{equation}
where $\pi$ is any permutation on $\gf_2^t$ and $g$ is any Boolean function on $\gf_2^t$.  The dual function of the MM class $f$ is
\begin{equation}
	\label{MM_dual}
	f^{*}(x,y) = y \cdot \pi^{-1}(x) + g(\pi^{-1}(x)).
\end{equation}

{In the final of this subsection, we provide some formulas about computing the Walsh transform of some Boolean functions from Bent functions, which will be used in Section \ref{Constricton_from_vectorial_Boolean}.

\begin{Lemma}
	\label{Wal_f1f2}
	Let $f_1$ and $f_2$ be two different Bent functions from $\gf_2^m$ to $\gf_2$ with $m=2t$. Let $f=f_1f_2$. Then 
	$$W_f(b) = \left\{ \begin{array}{ll}
		2^{m-1} + 2^{t-1}\left( (-1)^{f^{*}_1(0)}+(-1)^{f^{*}_2(0)}-(-1)^{f^{*}_3(0)} \right), & \mathrm{if}~b=0, \\
		2^{t-1}\left((-1)^{f^{*}_1(b)}+(-1)^{f^{*}_2(b)}-(-1)^{f^{*}_3(b)}\right), & \mathrm{if}~b\ne0,
	\end{array} \right. $$ 
	where $f_3=f_1+f_2$.
\end{Lemma}
\begin{proof}
	For any given $b\in\gf_{2^m}$, we have 
		\begin{eqnarray*}
		W_{f}(b)&=&\sum\limits_{x\in \mathbb{F}_{2^m}}(-1)^{f_1(x)f_2(x)+\mathrm{Tr}_{2^m}(b x)}\nonumber\\
		&=&\sum\limits_{x \in \mathbb{F}_{2^m}}\sum\limits_{\mathrm{i} \in \mathbb{F}_{2}}(-1)^{\mathrm{i}f_1(x)+\mathrm{Tr}_{2^m}(b x)}\frac{1+(-1)^{\mathrm{i}+f_2(x)}}{2}\nonumber\\
		&=&\frac{1}{2}\left(\sum\limits_{x\in \mathbb{F}_{2^m}}(-1)^{\mathrm{Tr}_{2^m}(b x)}+W_{f_1}(b)+W_{f_2}(b)-W_{f_3}(b)\right)\nonumber\\
		&=&\left\{ \begin{array}{ll}
			2^{m-1} + 2^{t-1}\left( (-1)^{f^{*}_1(0)}+(-1)^{f^{*}_2(0)}-(-1)^{f^{*}_3(0)} \right), & \mathrm{if}~b=0, \\
			2^{t-1}\left((-1)^{f^{*}_1(b)}+(-1)^{f^{*}_2(b)}-(-1)^{f^{*}_3(b)}\right), & \mathrm{if}~b\ne0.
		\end{array} \right.
	\end{eqnarray*}
\end{proof}

By the above lemma, we can easily get the following corollary.

\begin{Cor}
	\label{Wal_f1f2+f1+f2}
	Let  $f_1$ and $f_2$ be two different Bent functions from $\gf_2^m$ to $\gf_2$ with $m=2t$.  Let $f=f_1f_2+f_1+f_2$. Then
	$$W_f(b) = \left\{ \begin{array}{ll}
	 -	2^{m-1} + 2^{t-1}\left( (-1)^{f^{*}_1(0)}+(-1)^{f^{*}_2(0)}+(-1)^{f^{*}_3(0)} \right), & \mathrm{if}~b=0, \\
		2^{t-1}\left((-1)^{f^{*}_1(b)}+(-1)^{f^{*}_2(b)}+(-1)^{f^{*}_3(b)}\right), & \mathrm{if}~b\ne0,
	\end{array} \right. $$ 
	where $f_3=f_1+f_2$.
\end{Cor}
\begin{proof}
	In this corollary, $f=(f_1+1)(f_2+1)+1$. Then the formula of the Walsh transform of $f$ can be directly obtained by Lemma \ref{Wal_f1f2} and Remark \ref{f_f+1}.
\end{proof}
}

\subsection{Basic knowledge of binary linear codes}

In the following, some basic knowledge of binary linear codes is recalled.
Since a linear code is a subspace, all its codewords can be described in terms of a basis. That is to say, knowing a basis for a linear code enables us to describe its codewords explicitly.  In coding theory, a basis for a linear code is often represented in the form of a matrix, called a generator matrix, while a matrix that represents a basis for the dual code is called a parity-check matrix.

\begin{Def}
\cite{ling2004coding}
    \begin{enumerate}[(1)]
        \item A generator matrix for a linear code $\C$ is a matrix $G$ whose rows form a basis for $\C$.
        \item A parity-check matrix $H$ for a linear code $\C$ is a generator matrix for the dual code $\C^{\perp}$.
    \end{enumerate}
\end{Def}



In the literature, there are two general approaches to constructing binary linear codes from (vectorial) Boolean functions, see e.g. \cite{ding2016construction,ding2018minimal}.

Let $F$ be a vectorial Boolean function from $\gf_2^m$ to $\gf_2^k$ such that $F(\mathbf{0}_m)=\mathbf{0}_k$ and $a\cdot F(x)$ is not linear for any $a\in\gf_2^k\backslash \{\mathbf{0} \}$. We now define a linear code by
\begin{equation}
    \label{first_method}
    \C_F = \left\{ \bu_{a,b} = (a \cdot F(x) + b\cdot x)_{x\in\gf_2^m\backslash\{\mathbf{0} \} }: a\in\gf_2^k, b\in\gf_2^m   \right\}.
\end{equation}

 The following result given in \cite{li2022minimal} provides the parameters of $\C_F$.

\begin{Lemma}
  \cite{li2022minimal}
  \label{li_method}
      Let $F$ be a $(m,k)$-function such that $F(\mathbf{0}_m)=\mathbf{0}_k$ and $a\cdot F(x)$ is not linear for any $a\in\gf_2^k\backslash \{\mathbf{0} \}$. Then the binary linear code $\C_F$ derived in \eqref{first_method} has length $2^m-1$ and dimension $m+k$. Moreover,  for any $\bu_{a,b}\in \C_F$, we have
      $$\wt(\bu_{a,b}) = \begin{cases}
      	0, & ~\text{if}~ a = \mathbf{0}_k, b = \mathbf{0}_m \\
      	2^{m-1}, & ~\text{if}~ a = \mathbf{0}_k, b \in \gf_2^m\backslash\{\mathbf{0}_m\} \\
      	2^{m-1} - \frac{1}{2} W_F(a,b), &~\text{otherwise.}
      \end{cases}$$
\end{Lemma}

The above lemma is actually a generalization of the following result, which considers the Boolean case, i.e., $k=1$.

\begin{Lemma}
	\cite{carlet2007nonlinearities,mesnager2017linear,ding2018minimal}
\label{ding_method}
   Let $f$ be a Boolean function from $\gf_2^m$ to $\gf_2$ with $f(\mathbf{0}_m)=0$. The binary linear code $\C_f$ in \eqref{first_method} has length $2^m-1$ and dimension $m+1$ if $f(x)\neq \omega \cdot x$ for all $\omega\in\gf_2^m$. In addition, the weight distribution of $\C_f$ is given by the following multiset union:
    $$    \left\{ \frac{2^m-W_f(\omega)}{2}: \omega\in\gf_2^m  \right\} \bigcup \left\{ 2^{m-1}: \omega\in\gf_2^m\backslash\{ \mathbf{0}\} \right\}\bigcup \{0\}.$$
\end{Lemma}


There is another general approach to constructing linear codes from Boolean functions. Let $f$ be a Boolean function from $\gf_2^m$ to $\gf_2$ and $D_f = \{ x: x \in \gf_2^m ~|~ f(x) = 1  \} = \{ d_1,d_2,\ldots, d_n \}$. We define a binary linear code of length $n$  by
\begin{equation}
    \label{second_method}
    \C_{D_f} = \left\{ (x\cdot d_1, x\cdot d_2, \ldots, x\cdot d_n): x \in \gf_2^m \right\}.
\end{equation}
The dimension and weight distribution of $\C_{D_f}$ are determined by the following lemma.

\begin{Lemma}
\cite{ding2016construction}
  Let $f$ be a Boolean function from $\gf_2^m$ to $\gf_2$ and $D_f = \{ x: x \in \gf_2^m ~|~ f(x) = 1  \}$. If   $2n+W_f(\omega)\neq0$   for all $\omega\in\gf_2^m$, then the binary linear code $\C_{D_f}$ in \eqref{second_method} has length $n$ and dimension $m$, and its weight distribution is given by the following multiset:
  $$\left\{ \frac{2n+W_f(\omega)}{4}: \omega \in\gf_2^m\backslash\{\mathbf{0}\}     \right\} \bigcup \{0\}. $$
\end{Lemma}


We now recall the definition of minimal linear codes. In particular, Aschikhmin and Barg \cite{ashikhmin2002minimal} gave a sufficient condition for a linear code to be minimal, which is called the AB condition.

\begin{Def}
    Let $\C$ be a binary linear code. A codeword $\bu\in\C$ covers a codeword $\bv\in\C$ if $\supp(\bu)$ contains $\supp(\bv)$.  A codeword $\bu\in\C$ is called minimal if $\bu$ covers only the codewords $\bu$ and $\mathbf{0}$, but no other codewords in $\C$. $\C$ is said to be minimal if every codeword in $\C$ is minimal. Moreover, if $\C$ is minimal and $\frac{w_{\min}}{w_{\max}}\le \frac{1}{2}$, where $w_{\min}$ and $w_{\max}$ denote the minimum and maximum nonzero weights in $\C$, respectively, then we say $\C$ is a minimal linear code violating Aschikhmin-Barg (AB) condition.
\end{Def}



In 2018, Ding et al. \cite{ding2018minimal} provided a necessary and sufficient condition for a binary linear code to be minimal
as follows.

\begin{Lemma}
\cite{ding2018minimal}
\label{minimal}
        Let $\C$ be a binary linear code. Then $\C$ is minimal if and only if for any distinct nonzero codewords $\bu,\bv\in\C$, we have
        \begin{equation}
            \label{minimal_condition} \wt(\bu+\bv)\neq \wt(\bu)-\wt(\bv).
        \end{equation}
        Moreover, if $\C$ is two-weight with weights $w_1$ and $w_2$, where $0<w_1<w_2$ and $w_2\neq 2w_1$, then $\C$ is minimal.
\end{Lemma}

 Furthermore, by Lemma \ref{minimal}, Ding et al. \cite{ding2018minimal}  provided a necessary and sufficient condition for Boolean functions $f$ such that the linear codes $\C_f$ defined in \eqref{first_method} are minimal.

\begin{Lemma}
    \cite{ding2018minimal}
    \label{f_minimal}
    Let $\C_f$ be the linear code in \eqref{first_method}. Then $\C_f$ is minimal if and only if
    $$W_f(h)+W_f(l)\neq 2^m~\text{and}~W_f(h)-W_f(l)\neq 2^m$$
    for every pair of distinct vector elements $h$ and $l$ in $\gf_2^m$.
\end{Lemma}

Later in \cite{li2022minimal}, Li et al. gave a complete characterization for vectorial Boolean functions $F$ such that the linear codes $\C_F$ derived in \eqref{first_method} are minimal.

\begin{Lemma}
    \cite{li2022minimal}
    \label{li_minimal}
      Let $F$ be a $(m,k)$-function such that $F(\mathbf{0}_m)=\mathbf{0}_k$ and $a\cdot F(x)$ is not linear for any $a\in\gf_2^k\backslash \{\mathbf{0} \}$. Then the binary linear code $\C_F$ derived in \eqref{first_method} is minimal if and only if the following two conditions are satisfied.
    \begin{enumerate}[(1)]
    	\item For  any $a\in\gf_2^k\backslash\{\mathbf{0}\}$ and $b,b^{'}\in\gf_2^m$ with $b\neq b^{'}$, it holds that
    	$$W_F(a,b)\pm W_F(a,b^{'}) \neq 2^m.$$
    	\item For any $a,a^{'}\in\gf_2^k\backslash\{\mathbf{0}\}$ with $a\neq a^{'}$, and any $b,b^{'}\in\gf_2^m$, it holds that
    	$$W_F(a,b)+W_F(a^{'},b^{'}) - W_F(a+a^{'}, b+b^{'}) \neq 2^m.$$
    \end{enumerate}
\end{Lemma}

\section{Some characterizations of binary SO (singly-even) linear codes}

\label{Characterization}
In this section, we mainly provide some characterizations of binary SO (singly-even) linear codes. In particular, a necessary and sufficient condition for vectorial Boolean functions $F$ such that the linear codes $\C_F$ defined in \eqref{first_method} are SO is given.

 The following lemma is a characterization of binary SO linear codes, although simple but powerful.

\begin{Lemma}
    \label{main_th}
    Let $\C$ be a binary linear code. Then $\C$ is SO if and only if for any $\bu, \bv \in \C$, we have
    \begin{equation}
        \label{main_eq}
        \wt(\bu) + \wt(\bv) \equiv \wt(\bu+\bv) \pmod 4.
    \end{equation}
\end{Lemma}

\begin{proof}
    According to the definition of SO codes, $\C$ is SO if and only if for any $\bu, \bv\in \C$,   $\bu\cdot \bv = 0$. That is to say,
    $ \# \left( \supp(\bu) \cap \supp(\bv)  \right) $
    must be even. By \cite[Theorem 1.4.3]{huffman2010fundamentals}, we have that for any $\bu,\bv \in \C$,
   $$\wt(\bu) + \wt(\bv) = \wt(\bu+\bv) + 2 \# \left( \supp(\bu) \cap \supp(\bv)  \right). $$
    Thus $\C$ is SO if and only if $\wt(\bu) + \wt(\bv) -  \wt(\bu+\bv)$ is divisible by four, i.e.,
    $$\wt(\bu) + \wt(\bv) \equiv \wt(\bu+\bv) \pmod 4.$$
\end{proof}

By Lemma \ref{main_th}, Theorem \ref{basic_property} is trivial.  Moreover, we can obtain more properties of binary SO linear codes. Particularly, we give a necessary and sufficient condition for binary  SO singly-even linear codes.

\begin{Prop}
\label{simple_lemma}
    Let $\C$ be a binary $[n,k,d]$-linear code.
    \begin{enumerate}[(1)]
        \item If $\C$ is SO, singly-even and $\#\C>2$, then there exist some codewords in $\C$ that have weights divisible by four.
        \item Let $\C_1 = \{\bu\in\C: \wt(\bu)\equiv 0\pmod 4\}$. Then $\C$ is  SO and singly-even if and only if $\C_1$ is an $[n,k-1]$-binary linear code and the weights of all codewords in $\C$ are even.
        \item If $\{\bu_1,\bu_2,\ldots,\bu_k\}$ is a basis of $\C$, then $\C$ is SO if and only if for any $1\le i,j\le k$,
    $$ \wt(\bu_i) + \wt(\bu_j) \equiv \wt(\bu_i+\bu_j) \pmod 4.$$
    \end{enumerate}
\end{Prop}

\begin{proof}
 (1) The statement holds since for any two distinct codewords $\bu,\bv\in \mathcal{C}$ with $\wt(\bu)\equiv 2 \pmod 4$ and $\wt(\bv)\equiv 2 \pmod 4$, we have
 $$\wt(\bu)+\wt(\bv)\equiv \wt(\bu+\bv)\pmod 4,$$
 where the first equality is due to Lemma \ref{main_th}.

 (2) $\Rightarrow$): It is trivial that the weights of all codewords in $\C$ are even. We now let $f$ be a function from the code $\mathcal{C}$ to $\mathbb{Z}_4$ defined by
 $$f(\bu)=u_1+u_2+\cdots +u_n\pmod 4,$$
 where $\bu=(u_1,u_2,\cdots,u_n)\in\mathcal{C}$. Lemma \ref{main_th} implies $f$ is a homomorphism. It then follows from the definition of $f$ that
 $$\ker(f)=\mathcal{C}_{1}=\{\bu\in\mathcal{C}:~ \wt(\bu)\equiv 0 \pmod 4\}.$$
 Furthermore, it follows from the isomorphism $\bar{f} : \mathcal{C}/\ker(f)\rightarrow \{0,2\}$ that $\mathcal{C}_{1}$ is a $[n, k-1]$ binary linear code.

 $\Leftarrow$): Let $\mathcal{C}_{2}=\{\bu\in\mathcal{C}:~ \wt(\bu)\equiv 2 \pmod 4\}$. Then $\C=\C_1\cup\C_2$. Since $\mathcal{C}_{1}$ is a $[n, k-1]$ binary linear code, it follows that for any $\bu\in\mathcal{C}_1$ and $\bv\in\mathcal{C}_2$,
 $$\bv+\mathcal{C}_{1}=\mathcal{C}_2, \bv+\mathcal{C}_{2}=\mathcal{C}_1,  \bu+\mathcal{C}_2=\mathcal{C}_2,$$
 implying that for any $\bu,\bv\in \mathcal{C}$, the following equality
 $$\wt(\bu)+\wt(\bv)\equiv \wt(\bu+\bv)\pmod 4$$
 holds.

   (3) $\Rightarrow):$ Trivial.

    $\Leftarrow):$ for any $\bu,\bv\in\C$, there exist some $a_1,a_2,\ldots,a_k,b_1,b_2,\ldots,b_k\in\gf_2$ such that
    $$\bu = a_1\bu_1 + a_2 \bu_2 + \cdots + a_k \bu_k, \bv = b_1\bu_1 + b_2 \bu_2 + \cdots + b_k \bu_k. $$
    Then
    \begin{eqnarray*}
        \wt(\bu) &=& \wt ( a_1\bu_1 + a_2 \bu_2 + \cdots + a_k \bu_k) \\
        &\equiv& a_1 \wt(\bu_1) + a_2 \wt(\bu_2) + \cdots + a_k \wt(\bu_k) \pmod 4
    \end{eqnarray*}
    since for any $1\le i,j\le k$,
    $ \wt(\bu_i) + \wt(\bu_j) \equiv \wt(\bu_i+\bu_j) \pmod 4.$  Similarly, we have
    $$\wt(\bv) \equiv b_1 \wt(\bu_1) + b_2 \wt(\bu_2) + \cdots + b_k \wt(\bu_k) \pmod 4$$
    and
    $$\wt(\bu+\bv) \equiv (a_1+b_1) \wt(\bu_1) + (a_2+b_2) \wt(\bu_2) + \cdots + (a_k+b_k) \wt(\bu_k) \pmod 4.$$
    Hence, $\wt(\bu+\bv) \equiv \wt(\bu) + \wt (\bv) \pmod 4$ holds and then $\C$ is SO by Lemma \ref{main_th}.
\end{proof}

\begin{Rem}
 The   sufficiency of  Proposition \ref{simple_lemma} (2) has been proved in \cite[Theorem 1.4.6]{huffman2010fundamentals}. Our main contribution is the necessity.
\end{Rem}

In \cite{jin2025several}, the authors gave a sufficient and necessary condition for Boolean functions $f$ such that $\C_f$ derived in \eqref{first_method} is SO.

\begin{Lemma}
	\label{Cf_SO}
	\cite{jin2025several}
	Let $m\ge 3$, $f$ be a Boolean function from $\gf_2^m$ to $\gf_2$ such that $f(\mathbf{0})=0$ but $f(b)=1$ for at least one $b\in\gf_2^m$. Recall that $$\C_f = \left\{\bu_{a,b} = (af(x)+b\cdot x)_{x\in\gf_2^m\backslash\{ \mathbf{0} \}}: a\in\gf_2,b\in\gf_2^m  \right\}.$$
	Then $\C_f$ is SO if and only if $W_f(b)\pm W_f(\mathbf{0})\equiv 0\pmod 8$ for any $b\in\gf_2^m$.
\end{Lemma}

On the basis of Lemma \ref{Cf_SO}, we now provide a complete characterization for Boolean functions $f$ such that $\C_f$ is SO singly-even. Some explicit constructions will be given in Section \ref{Constructions_from_Boolean}.

\begin{Lemma}
	\label{Cf_SO_SE}
	Let $m\ge 3$, $f$ be a Boolean function from $\gf_2^m$ to $\gf_2$ such that $f(\mathbf{0})=0$ but $f(b)=1$ for at least one $b\in\gf_2^m$. Recall that $$\C_f = \left\{\bu_{a,b}=(af(x)+b\cdot x)_{x\in\gf_2^m\backslash\{ \mathbf{0} \}}: a\in\gf_2,b\in\gf_2^m  \right\}.$$  Then $\C_f$ is a binary SO singly-even linear code if and only if $W_f(b)\equiv 4\pmod 8$ for any $b\in\gf_2^m$.
\end{Lemma}

\begin{proof}
	If $W_f(b)\equiv 4\pmod 8$ for any $b\in\gf_2^m$, then it is clear that  $W_f(b)\pm W_f(\mathbf{0})\equiv 0\pmod 8$ for any $b\in\gf_2^m$ and then $\C_f$ is SO by Lemma \ref{Cf_SO}. Moreover, in this case, $\wt(\bu_{1,b}) = \frac{2^m-W_f(b)}{2} \equiv 2\pmod 4$ for any $b\in\gf_2^m$. That is to say, $\C$ is singly-even.
	
	On the contrary, if $\C_f$ is singly-even, then there exists at least one element $b_0\in\gf_2^m$ such that $\wt(\bu_{1,b_0}) = \frac{2^m-W_f(b_0)}{2}\equiv 2\pmod 4$, i.e., $W_f(b_0)\equiv 4\pmod 8$. Moreover, since $\C_f$ is SO, by Lemma \ref{Cf_SO}, $W_f(b_0)\pm W_f(\mathbf{0}) \equiv 0\pmod 8$ and then $W_f(\mathbf{0})\equiv 4 \pmod 8$. Thus for all $b\in\gf_2^m$, $W_f(b)\equiv 4\pmod 8.$
\end{proof}

{
{In the final of this section, we generalize  Lemma \ref{Cf_SO} by considering a more general case where $F$ is a vectorial Boolean function. Some explicit constructions will be given in Section \ref{Constricton_from_vectorial_Boolean}.}

\begin{Lemma}
	\label{CF_SO}
	Let $m$, $k$ be positive integers and $m\geq 3$. Let $F$ be a vectorial Boolean function from $\mathbb{F}_2^m$ to $\mathbb{F}_2^k$. Then the $\mathcal{C}_F$ derived in \eqref{first_method} is a SO linear code if and only if $W_F(a_1, b_1)+W_F(a_2, b_2) - W_F(a_1+a_2, b_1+b_2) \equiv 0 \pmod 8$, for any $a_1,a_2 \in\mathbb{F}_2^k$, and $b_1,b_2\in\mathbb{F}_2^m$.
\end{Lemma}

\begin{proof}
	According to Lemma \ref{main_th}, $\C_F$ is SO if and only if for all $a_1,a_2 \in\mathbb{F}_2^k$, and $b_1,b_2\in\mathbb{F}_2^m$,
	\begin{equation}
		\label{CF_SO_eq_1}
		\wt(\bu_{a_1,b_1}) + \wt(\bu_{a_2,b_2}) \equiv \wt(\bu_{a_1+a_2,b_1+b_2}) \pmod 4.
	\end{equation}
	
	By Lemma \ref{li_method},  since $\wt(\bu_{\mathbf{0}_k, b}) \equiv 0 \pmod 4$ and $\wt(\bu_{a,b}) \equiv -\frac{1}{2} W_F(a,b)   \pmod 4$ for all $a\neq \mathbf{0}_k$ and $b\in\gf_2^m$, Eq. \eqref{CF_SO_eq_1} is equivalent to
	$$W_F(a_1, b_1)+W_F(a_2, b_2) - W_F(a_1+a_2, b_1+b_2) \equiv 0 \pmod 8$$ for any $a_1,a_2 \in\mathbb{F}_2^k$, and $b_1,b_2\in\mathbb{F}_2^m$.
\end{proof}

\begin{Rem}
    It is clear that Lemma \ref{Cf_SO} is a special case of Lemma \ref{CF_SO}. In fact, in Lemma \ref{CF_SO}, let $k=1$. Then the vectorial Boolean function $F$ is a Boolean one $f$ and $\C_f$ is SO if and only if (I) $W_f(b_2)-W_f(b_1+b_2)\equiv 0\pmod 8$ ($a_1=0,a_2=1$) and (II) $W_f(b_1)+W_f(b_2)\equiv 0\pmod 8$ ($a_1=a_2=1$) for any $b_1,b_2\in\mathbb{F}_2^m$. In (I), if $b_2=\mathbf{0}$, then we have $W_f(\mathbf{0})-W_f(b_1)\equiv 0\pmod 8$. In (II), if $b_2=\mathbf{0}$, then we have $W_f(\mathbf{0})+W_f(b_1)\equiv 0\pmod 8$. Thus, if $\C_f$ is SO, then we have $W_f(b)\pm W_f(\mathbf{0})\equiv 0\pmod 8$ for any $b\in\gf_2^m$. On the contrary, if $W_f(b)\pm W_f(\mathbf{0})\equiv 0\pmod 8$ for any $b\in\gf_2^m$, then $W_f(b_2) \equiv W_f(\mathbf{0})\equiv W_f(b_1+b_2) \pmod 8$, implying $W_f(b_2)-W_f(b_1+b_2)\equiv 0\pmod 8$, i.e., (I) holds for  any $b_1,b_2\in\mathbb{F}_2^m$. Similarly, (II) also holds. Therefore, $\C_f$ is SO if and only if $W_f(b)\pm W_f(\mathbf{0})\equiv 0\pmod 8$ for any $b\in\gf_2^m$, i.e., Lemma \ref{Cf_SO} holds.
\end{Rem}
}

\section{Constructions of binary SO singly-even minimal linear codes from Boolean functions}

\label{Constructions_from_Boolean}




In the following, we propose some explicit constructions of binary SO singly-even minimal linear codes violating the AB condition with few weights from Boolean functions.

In \cite[Theorem 18]{ding2018minimal}, Ding et al. used the partial spread to construct a family of binary minimal linear codes violating the AB condition. Let $m$ be a positive even integer with $m\ge 6$ and $t=\frac{m}{2}$. A partial spread of order $s$ ($s$-spread) in $\gf_2^m$ is a set of $s$ $t$-dimensional subspaces $E_1,E_2,\ldots,E_s$ of $\gf_2^m$ such that $E_i\cap E_j = \{\mathbf{0} \}$ for all $1\le i<j\le s$. Let $f_i: \gf_2^m \to \gf_2$ be the Boolean function with support $E_i\backslash\{\mathbf{0}\}$, i.e.,
$$f_i(x) = \begin{cases}
	1, & \text{if}~ x\in E_i\backslash\{\mathbf{0}\} \\
	0, & \text{otherwise}
\end{cases}$$
and
\begin{equation}
	\label{sum_f_ix}
	f(x) = \sum_{i=1}^{s} f_i(x).
\end{equation}
Then Ding et al. obtained the following theorem.
\begin{Th}
	\cite[Theorem 18]{ding2018minimal}
	\label{ding}
	Let $s$ be any integer with $1\le s\le 2^t+1$ and $s\notin\{ 1,2^t,2^t+1 \}$, and $f$ be the Boolean function defined in \eqref{sum_f_ix}. Then the code $\C_f = \left\{(af(x)+b\cdot x)_{x\in\gf_2^m\backslash\{ \mathbf{0} \}}: a\in\gf_2,b\in\gf_2^m  \right\}$ is a minimal $[2^m-1, m+1]$-linear code. Furthermore, when $s\le 2^{t-2}$, $\frac{w_{\min}}{w_{\max}}\le \frac{1}{2}$.
\end{Th}

In the following, we study the SO singly-even property of codes in Theorem \ref{ding}.

\begin{Th}
	\label{ding_SO_SE}
	Let all notions be defined in Theorem \ref{ding}. Then $\C_f$ is SO singly-even if and only if $s\equiv 2\pmod 4$. Thus, when $s\equiv 2\pmod 4$ and $s\le 2^{t-2}$, $\C_f$ is a class of binary SO singly-even minimal linear codes violating the AB condition. 
\end{Th}

\begin{proof}
	According to \cite[Lemma 16]{ding2018minimal}, we have
	$$W_f(b) = \begin{cases}
		2^m-2s(2^t-1), & ~\text{if}~ b=\mathbf{0} \\
		2s, &  ~\text{if}~ b\notin E_i^{\perp} ~\text{for all}~1\le i\le s \\
		-2^{t+1}+2s, & ~\text{if}~ b\neq\mathbf{0} ~\text{and}~b\in E_i^{\perp}~\text{for some}~ 1\le i\le s,
	\end{cases}$$
	where $E_i^{\perp}$ denotes the dual space of $E_i$. According to Lemma \ref{Cf_SO_SE}, $\C_f$ is SO singly-even if and only if $W_f(b)\equiv 4 \pmod 8$ for all $b\in\gf_2^m$. Thus  $\C_f$ is SO singly-even if and only if $2s\equiv 4\pmod 8$, i.e., $s\equiv 2\pmod 4$. The last statement holds due to the above proof and Theorem \ref{ding}. 
\end{proof}

\begin{example}
	Following \cite[Example 20]{ding2018minimal}, let $m=6$ and $s=2$. Then $\C_f$ in Theorem \ref{ding} is a minimal  $[63,7,14]$-linear code with the weight enumerator $1+z^{14}+49z^{30}+63z^{32}+14z^{38}$. Moreover, by MAGMA, $\C_f$ is SO singly-even, which is consistent with Theorem \ref{ding_SO_SE}.
\end{example}

Since the appearance of the paper \cite{ding2018minimal} of Ding et al., a large number of minimal linear codes that violate the AB condition have been constructed, see e.g. \cite{mesnager2019several,mesnager2023several,zhang2019four,li2020four,li2022minimal}. We believe that using our method, one can easily determine whether these codes are SO singly-even or not.

{We now provide a general approach to constructing Boolean functions $f$ over $\gf_{2^m}$ satisfying $W_f(b)\equiv 4 \pmod 8$ for all $b\in\gf_{2^m}$.
	
	\begin{Lemma}
		\label{lemma_key}
		Let $g$ be a Boolean function from $\gf_2^k$ to $\gf_2$ satisfying $W_g(b)\equiv 0\pmod 4$ for all $b\in\gf_2^k$. Suppose $\alpha\in\gf_2^k$ satisfying $g(\alpha)=0$ and
		\begin{equation}
			\label{f48}
			f(x,y_1,y_2) = g_{\alpha}(x) + y_1y_2
		\end{equation}
		be a Boolean function from $\gf_2^m$ to $\gf_2$, where $m=k+2$, $(y_1,y_2)\in\gf_2^2$ and $$g_{\alpha}(x) = \begin{cases}
			g(x), &~\text{if}~x\in\gf_2^k\backslash\{\alpha\} \\
			1, &~\text{if}~x = \alpha.
		\end{cases}$$
		Then for any $(b,c_1,c_2)\in\gf_2^m$,
		\begin{equation}
			\label{key_eq}
			W_f(b,c_1,c_2) = \begin{cases}
				-2 \left( W_g(b)-2  (-1)^{b\cdot \alpha} \right), & ~\text{if}~(c_1,c_2)=(1,1) \\
				2  \left( W_g(b)-2  (-1)^{b\cdot \alpha} \right), & ~\text{otherwise.}
			\end{cases}
		\end{equation}
		and thus  $W_f(b,c_1,c_2)\equiv 4\pmod 8$.
	\end{Lemma}
	
	\begin{proof}
		For any $(b,c_1,c_2)\in\gf_2^m$, we have
		\begin{eqnarray*}
			W_f(b,c_1,c_2) &=& \sum_{x\in\gf_2^k, y_1,y_2\in\gf_2}(-1)^{g_{\alpha}(x)+y_1y_2+b\cdot x + c_1y_1+c_2y_2} \\
			&=& \sum_{y_1,y_2\in\gf_2} (-1)^{y_1y_2+c_1y_1+c_2y_2} \sum_{x\in\gf_2^k} (-1)^{g_{\alpha}(x)+b\cdot x}.
		\end{eqnarray*}
		Note that it is easy to check that for any $(c_1,c_2)\in\gf_2^2$,
		$$\sum_{y_1,y_2\in\gf_2}(-1)^{y_1y_2+c_1y_1+c_2y_2} = \begin{cases}
			-2, & ~\text{if}~(c_1,c_2)=(1,1)\\
			2, & ~\text{otherwise.}
		\end{cases}$$
		Moreover, for any $b\in\gf_2^k$,
		\begin{eqnarray*}
			& & \sum_{x\in\gf_2^k}(-1)^{g_{\alpha}(x)+b\cdot x} \\
			&=& \sum_{x\in\gf_2^k\backslash\{\alpha\}}(-1)^{g(x)+b\cdot x} - (-1)^{b\cdot \alpha} \\
			&=& \sum_{x\in\gf_2^k} (-1)^{g(x)+b\cdot x} - 2  (-1)^{b \cdot \alpha} = W_g(b)-2  (-1)^{b\cdot \alpha}.
		\end{eqnarray*}
		Therefore, we get
		$$ W_f(b,c_1,c_2) = \begin{cases}
			-2 \left( W_g(b)-2  (-1)^{b\cdot \alpha} \right), & ~\text{if}~(c_1,c_2)=(1,1) \\
			2  \left( W_g(b)-2  (-1)^{b\cdot \alpha} \right), & ~\text{otherwise.}   \end{cases} $$
		Since for all  $b\in\gf_2^k$, $W_g(b)\equiv 0\pmod 4$, we can see that for all $(b,c_1,c_2)\in\gf_2^m$,
		$$W_f(b,c_1,c_2)\equiv 4 \pmod 8.$$
	\end{proof}
	
	By the above lemma, we provide a special case, where $g(\mathbf{0})=0$ and $\alpha=\mathbf{0}$. However, in this case, $f(\mathbf{0})=1$, which does not satisfy the condition of Lemma \ref{ding_method}. Thus, we give a slight modification as follows.
	
	\begin{Cor}
		\label{cor}
		Let $g$ be a Boolean function from $\gf_2^k$ to $\gf_2$ satisfying $g(\mathbf{0})=0$ and $W_g(b)\equiv 0\pmod 4$ for all $b\in\gf_2^k$. Let
		\begin{equation}
			\label{slight_modification}
			f(x,y_1,y_2) = g_0(x)+y_1y_2+1
		\end{equation}
		be a Boolean function from $\gf_2^m$ to $\gf_2$, where $m=k+2, (y_1,y_2)\in\gf_2^2$ and $g_0(x)$ is defined as in Lemma \ref{lemma_key}. Then for any  $(b,c_1,c_2)\in\gf_2^m$,
		\begin{equation}
			\label{key_eq1}
			W_f(b,c_1,c_2) = \begin{cases}
				2 \left( W_g(b)-2   \right), & ~\text{if}~(c_1,c_2)=(1,1) \\
				-2  \left( W_g(b)-2  \right), & ~\text{otherwise.}
			\end{cases}
		\end{equation}
	\end{Cor}
	
	It is well-known that there exists a large number of Boolean functions $g$ from $\gf_2^k$ to $\gf_2$ satisfying $W_g(b)\equiv 0\pmod 4$ for all $b\in\gf_2^k$, such as bent functions, semi-bent functions, and so on. In the following, we give an explicit construction from bent functions as an example.
	
	\begin{Th}
		\label{five_weight}
		Let $k$ be even, $g$ be a bent function from $\gf_2^k$ to $\gf_2$ with $g(\mathbf{0})=0$, $m=k+2$. Let $f$ be defined as in \eqref{slight_modification}. Then the code $\C_{f} = \left\{(af(x)+b\cdot x)_{x\in\gf_2^m\backslash\{ \mathbf{0} \}}: a\in\gf_2,b\in\gf_2^m  \right\}$ is a binary SO singly-even  $[2^{k+2}-1, k+3, 2^{k+1}-2^{\frac{k}{2}}-2]$-linear code with the weight distribution in Table \ref{WD_five_weight}.
		\begin{table}[h]
			\centering
			\caption{The weight distribution of the codes $\C_f$ in Theorem \ref{five_weight}}	\label{WD_five_weight}
			\begin{tabular}{cc}
				\hline
				Weight & Multiplicity \\
				\hline
				$0$ & $1$\\
				$2^{k+1} - 2^{\frac{k}{2}} -2$ & $3(2^{k-1}-2^{\frac{k}{2}-1})$  \\
				$2^{k+1}-2^{\frac{k}{2}}+2$ & $2^{k-1}+2^{\frac{k}{2}-1} $ \\
				$2^{k+1}$ & $2^{k+2}-1$ \\
				$2^{k+1}+2^{\frac{k}{2}}-2$ & $2^{k-1}-2^{\frac{k}{2}-1}$ \\
				$2^{k+1}+2^{\frac{k}{2}}+2$ & $3(2^{k-1}+2^{\frac{k}{2}-1})$  \\
				\hline
			\end{tabular}
		\end{table}
	\end{Th}
	\begin{proof}
		This proof can be easily obtained by Corollary \ref{cor}, Lemma  \ref{ding_method}, and the Walsh spectrum of bent functions.
	\end{proof}
	
	\begin{example}
		In Theorem \ref{five_weight}, let $k=6$, $\gamma$ be a primitive element of $\gf_{2^6}$, and $g(x)=\tr_{2^6}(\gamma x^3)$ be a bent function. Then, by MAGMA, we can get that $\C_f$ is a binary SO singly-even $[255,9,118]$-linear code with the weight enumerator $1+84z^{118}+36z^{122}+255z^{128}+108z^{134}+28z^{138}$, which is consistent with Theorem \ref{five_weight}.
		In addition, the dual code of $\C_f$ is a binary $[255,246,3]$-linear code, which is near-optimal by \cite{Grassl:codetables}.
	\end{example}

	Note that the linear code in Theorem \ref{five_weight} is minimal satisfying the AB condition. Finally, we propose an infinite class of binary linear codes that are SO, singly-even, minimal, violating the AB condition, and with few weights at the same time.
	
	\begin{Th}
		\label{five_weight_1}
		Let $g(x)=\tr_{2^k}(x)$, $m=k+2$ and $f$ be defined as in \eqref{slight_modification}. Then the code $$\C_{f} = \left\{(af(x)+b\cdot x)_{x\in\gf_2^m\backslash\{ \mathbf{0} \}}: a\in\gf_2,b\in\gf_2^m  \right\}$$ is a binary SO singly-even minimal $[2^{k+2}-1, k+3, 2^k+2]$-linear code violating the AB condition with the weight distribution in Table \ref{WD_five_weight_1}.
		\begin{table}[h]
			\centering
			\caption{The weight distribution of the codes $\C_f$ in Theorem \ref{five_weight_1}}	\label{WD_five_weight_1}
			\begin{tabular}{cc}
				\hline
				Weight & Multiplicity \\
				\hline
				$0$ & $1$\\
				$2^k+2$ & $1$  \\
				$2^{k+1}-2$ & $3(2^k-1)$ \\
				$2^{k+1}$ & $2^{k+2}-1$ \\
				$2^{k+1}+2$ & $2^k-1$ \\
				$3\cdot 2^k-2$ & $3$  \\
				\hline
			\end{tabular}
		\end{table}
	\end{Th}
	
	\begin{proof}
		The theorem can be proved together with Lemma \ref{ding_method}, Corollary \ref{cor}, Lemma \ref{f_minimal}, and the Walsh spectrum of $g$.
	\end{proof}
	
	\begin{example}
		In Theorem \ref{five_weight_1}, let $k=6$. Then by MAGMA, we can get the linear code $\C_f$ is a binary SO singly-even minimal $[255, 9, 66]$-linear code violating the AB condition with the weight enumerator $1+z^{66} + 189 z^{126} + 255 z^{128} + 63 z^{130} + 3 z^{190}$, which is consistent with Theorem \ref{five_weight_1}. Moreover, by MAGMA, the dual code of $\C_f$ is a binary $[255,246,3]$-linear code, which is near-optimal by \cite{Grassl:codetables}.
\end{example}}

{
\section{Constructions of binary SO singly-even minimal linear codes from vectorial Boolean functions}
\label{Constricton_from_vectorial_Boolean}

In this section, we construct some explicit classes of binary SO singly-even and minimal linear codes based on vectorial Boolean functions and Lemma \ref{CF_SO}. For that, we first introduce the conditions $(\mathbf{c}_1)$ and $(\mathbf{c}_2)$ for Boolean functions presented in \cite{li2022minimal}. 

\begin{Def}\cite{li2022minimal}
	\label{fgvc}
	Let $f$ be a Boolean function on $\mathbb{F}_{2^m}$ with $f(0)=0$, and satisfies the following two	conditions:
	$(\mathbf{c}_1)$ $W_{f}(b_1)\pm W_{f}(b_2)\ne 2^m$ for any $b_1,b_2\in\mathbb{F}_{2^m}$ and $b_1\ne b_2$;
	$(\mathbf{c}_2)$ $2\max_{b\in\mathbb{F}_{2^m}} W_{f}(b)-\min_{b\in\mathbb{F}_{2^m}}W_{f}(b)\ge 2^m$.
	
	Let $G$ be a $(m,k-1)$-function with $G(0)=0$ such that $vf+\tr_{2^{k-1}}(uG)$ is not affine for any
	nonzero $\mu=(v,u)\in\mathbb{F}_{2}\times\mathbb{F}_{2^{k-1}}$, and the linear code $\C_{G}$ derived in \eqref{first_method} is minimal.
\end{Def}

By Lemma \ref{li_minimal}, Li et al. provided the following approach to constructing binary minimal linear codes. 

\begin{Lemma}
	\cite{li2022minimal}\label{tccs} Let $F(x)=(f(x), G(x))$ be a $(m,k)$-function, where $f$ and $G$ are given in Definition \ref{fgvc}. Let $H_u(x)=f(x)+\tr_{2^{k-1}}(uG)$, where $u\in \mathbb{F}_{2^{k-1}}^{*}$. Then the linear code $\mathcal{C}_{F}$ derived in \eqref{first_method}  is minimal if and only if the following three conditions are fulfilled:
	\begin{enumerate}[(1)]
		\item for any $u\in \mathbb{F}_{2^{k-1}}^{*}$ and any $b_1, b_2\in\mathbb{F}_{2^m}$ with $b_1\ne b_2$, it holds that $W_{H_u}(b_1)\pm W_{H_u}(b_2)\ne 2^m$;
		\item for any $u\in\mathbb{F}_{2^{k-1}}^{*}$ and any $b_1, b_2\in\mathbb{F}_{2^m}$, it holds that $(-1)^{i}W_{f}(b_1)+(-1)^{j}W_{G}(u,b_2)+(-1)^{l}W_{H_u}(b_1+b_2)\ne 2^m$, where $i, j, l\in\{0,1\}$ such that exactly two of them are $0$;
		\item for any $u_1, u_2\in \mathbb{F}_{2^{k-1}}^{*}$ with $u_1\ne u_2$ and any $b_1, b_2\in\mathbb{F}_{2^m}$, it holds that
		$W_{H_{u_1}}(b_1)+(-1)^{i}W_{G}(u_2,b_2)+(-1)^{j} W_{H_{u_1+u_2}}(b_1+b_2)\ne 2^m$, where $i, j\in\{0,1\}$ such that exactly one of them is $0$.
	\end{enumerate}
\end{Lemma}

To design minimal linear codes violating the AB condition based on Lemma \ref{tccs}, below we provide a generic construction of Boolean functions satisfying the conditions $(\mathbf{c}_1)$ and $(\mathbf{c}_2)$  in Definition \ref{fgvc}.

\begin{Prop}\label{bvg-t}
	Let $f_1, f_2$ and $f_3=f_1+f_2$ be bent functions defined on $\mathbb{F}_{2^m}$ with $m=2t$. Then the following statements hold.
	\begin{enumerate}[(1)]
		\item If $f^{*}_1(0)= f^{*}_2(0)=f^{*}_3(0)=1$ and {there exist some elements $b\in\gf_{2^{m}}^{*}$ such that $f^{*}_1(b)=1, f^{*}_2(b) = 1, f^{*}_3(b) = 0$,} then the function $f=f_1f_2 $  satisfies the conditions $(\mathbf{c}_1)$ and $(\mathbf{c}_2)$ in Definition \ref{fgvc}.
		\item  If $f^{*}_1(0)= f^{*}_2(0)=f^{*}_3(0)=0$, then the function $f=f_1f_2 $  satisfies the conditions $(\mathbf{c}_1)$ and $(\mathbf{c}_2)$ in Definition \ref{fgvc}.
	\end{enumerate}
\end{Prop}

\begin{proof} By Lemma \ref{Wal_f1f2}, for any $b\in\gf_{2^m}$, we have 
	\begin{eqnarray}\label{vvlazero2}
	W_{f}(b) &=&\left\{ \begin{array}{ll}
			2^{m-1} + 2^{t-1}\left( (-1)^{f^{*}_1(0)}+(-1)^{f^{*}_2(0)}-(-1)^{f^{*}_3(0)} \right), & \mathrm{if}~b=0, \\
			2^{t-1}\left((-1)^{f^{*}_1(b)}+(-1)^{f^{*}_2(b)}-(-1)^{f^{*}_3(b)}\right), & \mathrm{if}~b\ne0.
		\end{array} \right.
	\end{eqnarray}
	
	We first prove the statement (1). It follows from the assumption that Eq. (\ref{vvlazero2}) can be written as
	\begin{eqnarray*}
		W_{f}(b)&=&\left\{ \begin{array}{ll}
			2^{m-1}-2^{t-1}, & \mathrm{if}~b=0, \\
			2^{t-1}\left((-1)^{f^{*}_1(b)}+(-1)^{f^{*}_2(b)}-(-1)^{f^{*}_3(b)}\right), & \mathrm{if}~b \ne0,
		\end{array} \right.
	\end{eqnarray*}
	which leads to
	\begin{eqnarray}\label{bvab}
		W_{f}(b)&=&\left\{ \begin{array}{ll}
			2^{m-1}-2^{t-1}, & \mathrm{if}~b =0, \\
			(-1)^{\mathrm{i}}2^{t-1}, & \mathrm{if}~ b\ne0,(f^{*}_1(b), f^{*}_2(b), f^{*}_3(b))\in\{(\mathrm{i},\mathrm{i},\mathrm{i}),\\&(1+\mathrm{i},\mathrm{i},1+\mathrm{i}) ,(\mathrm{i},1+\mathrm{i},1+\mathrm{i})\},\\
			(-1)^{\mathrm{i}}3\cdot2^{t-1}, & \mathrm{if}~b\ne0,(f^{*}_1(b), f^{*}_2(b), f^{*}_3(b))=(\mathrm{i},\mathrm{i},1+\mathrm{i}),
		\end{array} \right.
	\end{eqnarray}
	where $\mathrm{i}\in\{0,1\}$.
	
	By Eq. \eqref{bvab}, it is easy to see that $f=f_1f_2$ satisfies the conditions $(\mathbf{c}_1)$ and $(\mathbf{c}_2)$ in Definition \ref{fgvc} under the assumption that  $f^{*}_1(0)= f^{*}_2(0)=f^{*}_3(0)=1$ and and {there exist some elements $b\in\gf_{2^{m}}^{*}$ such that $f^{*}_1(b)=1, f^{*}_2(b) = 1, f^{*}_3(b) = 0$.}
	
	We are now proceed to prove the statement (2). Combining Eq. (\ref{vvlazero2}), we have  
	\begin{eqnarray*}\label{vbvab2}
		W_{f}(b)&=&\left\{ \begin{array}{ll}
			2^{m-1}+2^{t-1}, & \mathrm{if}~b=0, \\
			(-1)^{\mathrm{i}}2^{t-1}, & \mathrm{if}~b\ne0,(f^{*}_1(b), f^{*}_2(b), f^{*}_3(b))\in\{(\mathrm{i},\mathrm{i},\mathrm{i}),\\&(1+\mathrm{i},\mathrm{i},1+\mathrm{i}) ,(\mathrm{i},1+\mathrm{i},1+\mathrm{i})\}\\
			(-1)^{\mathrm{i}}3\cdot2^{t-1}, & \mathrm{if}~b\ne0,(f^{*}_1(b), f^{*}_2(b), f^{*}_3(b))=(\mathrm{i},\mathrm{i},1+\mathrm{i}),
		\end{array} \right.
	\end{eqnarray*}
	where $\mathrm{i}\in\{0,1\}$.	
	{Clearly, the only case violating the condition $(\mathbf{c}_2)$ is that for all $b\in\gf_{2^m}^{*}$, $(f^{*}_1(b), f^{*}_2(b), f^{*}_3(b))=(0,0,1)$, which cannot hold in this proposition due to the fact that $f_1^{*}$ is bent and then there exists some elements $b$ such that $f^{*}_1(b)=1$. Thus $f$ satisfies the conditions $(\mathbf{c}_1)$ and $(\mathbf{c}_2)$ in Definition \ref{fgvc}.}
\end{proof}

\begin{Rem}\label{bhgdt}
	By Lemma \ref{f_minimal}, we can easily deduce that the code constructed from Boolean functions, which satisfy the conditions $(\mathbf{c}_1)$ and $(\mathbf{c}_2)$ in Definition \ref{fgvc},  is minimal. This implies that Proposition \ref{bvg-t} provides a general construction of Boolean functions that can be used to design binary minimal linear codes.
\end{Rem}

In the rest of this section, we always assume that $m=2t$, $\tr(\cdot)=\tr_{2^t}(\cdot)$ for convenience. 
	
We now provide explicit constructions based on Proposition \ref{bvg-t} and a special subclass of the MM class.

\begin{Def}
	\label{fgF}
	Let $h_{\lambda} (x,y) = \tr (\lambda x\pi(y))$, where $\lambda\in\gf_{2^t}^{*}$ and $\pi$ is a linear permutation on $\gf_{2^t}$. Let $\lambda_1, \lambda_2$ be distinct elements in $\mathbb{F}_{2^t}^{*}$, and $w_1,w_2\in\mathbb{F}^*_{2^t}$ with $w_1\ne w_2$. Define some (vectorial) Boolean functions as follows. 
	\begin{align}
		f(x,y) = &~ h_{\lambda_1} (x,y)h_{\lambda_2} (x,y), \label{tf(x)1}\\
		g(x,y)= &~ h_{\lambda_3} (x,y)+((x+w_1)^{2^t-1}+(x+w_2)^{2^t-1})(y^{2^t-1}+1),\label{tg(x)}\\
		F(x,y)= &~ (f(x,y), g(x,y)) \label{tf(x)},
	\end{align}
	where $\lambda_3 = \lambda_1+\lambda_2$.
\end{Def}

 After plugging the vectorial Boolean function $F$ defined as in \eqref{tf(x)} into the linear code $\C_F$ derived in \eqref{first_method}, we can find that the resulting linear codes are SO and doubly-even. To obtain singly-even linear codes, we make a slight modification as follows.
\begin{equation}
	\label{C_F_modification}
	\C_F = \left\{ (  (a_1,a_2) \cdot F(x,y) + b_1x+b_2y )_{(x,y)\in\gf_2^t\times\gf_2^t\backslash\{ (\mathbf{0},\mathbf{0}) \} }: (a_1,a_2)\in\gf_2\times\gf_2, (b_1,b_2)\in \gf_{2^t}\times \V   \right\},
\end{equation}
where $\V=\{x\in\gf_{2^t} ~|~ \tr(xw_1)=\tr(xw_2)=0 \}$. It is trivial that $\V$ is a subspace of $\gf_{2^t}$ with the dimension $t-2$. That is to say, our code $\C_F$ is a subspace of that derived in \eqref{first_method}. Obviously, when the $b$ in Lemmas \ref{CF_SO} and \ref{tccs} belongs to a subspace of $\mathbb{F}_{2^m}$, the conclusions of these two lemmas still hold.

To compute the weight distribution of $\C_F$ derived in \eqref{C_F_modification}, it suffices to compute the partial Walsh transform of $W_F$, where $b_2$ belongs to $\V$, not $\gf_{2^t}$. We first compute the partial Walsh transforms of the Boolean functions $f$ and $g$ defined as in Definition \ref{fgF}.	

\begin{Lemma}
	\label{Wal_f_g_f+g}
	Let Boolean functions $f$ and $g$ be defined as in Definition \ref{fgF}. Let $\V=\{x\in\gf_{2^t} ~|~ \tr(xw_1)=\tr(xw_2)=0 \}$.  Then for any $(b_1,b_2)\in \gf_{2^t}\times \V$, we have 
	\begin{enumerate}[(1)]
		\item 
		$$W_f(b_1,b_2) = \begin{cases}
			2^{m-1}+2^{t-1}, &~\text{if}~ (b_1,b_2)=(0,0), \\
			2^{t-1} \left((-1)^{\A_1} + (-1)^{\A_2}-(-1)^{\A_3}\right), &~\text{if}~ (b_1,b_2)\neq(0,0);
		\end{cases}$$
		\item 
		$$W_g(b_1,b_2) = (-1)^{\A_3}2^t-4;$$
		\item  
			$$W_{f+g}(b_1,b_2) = \begin{cases}
			-2^{m-1}+3\cdot 2^{t-1}-4, &~\text{if}~ (b_1,b_2)=(0,0), \\
			2^{t-1} \left((-1)^{\A_1} + (-1)^{\A_2}+(-1)^{\A_3}\right)-4, &~\text{if}~ (b_1,b_2)\neq(0,0),
		\end{cases}$$
	\end{enumerate}
	where for all $i\in\{1,2,3\}$, $\A_i = \tr(b_2\pi^{-1}(\lambda_i^{-1}b_1))$. 
\end{Lemma}

\begin{proof}
	(1) By Lemma \ref{Wal_f1f2}, for any $(b_1,b_2)\in \gf_{2^t}\times \V$, we have 
	$$	W_f(b_1,b_2) = \left\{ \begin{array}{ll}
		2^{m-1} + 2^{t-1}\left( (-1)^{h_{\lambda_1}^{*}(0,0)}+(-1)^{h_{\lambda_2}^{*}(0,0)}-(-1)^{h_{\lambda_3}^{*}(0,0)} \right), & \text{if}~(b_1,b_2)=(0,0), \\
		2^{t-1}\left((-1)^{h_{\lambda_1}^{*}(b_1,b_2)}+(-1)^{h_{\lambda_2}^{*}(b_1,b_2)}-(-1)^{h_{\lambda_3}^{*}(b_1,b_2)}\right), & \text{if}~(b_1,b_2)\ne(0,0).
	\end{array} \right. $$
	Moreover, since $h_{\lambda_1}^{*}(0,0)=h_{\lambda_2}^{*}(0,0)=h_{\lambda_3}^{*}(0,0)=0$ and $h_{\lambda_i}^{*}(b_1,b_2) = \A_i$ for all $i\in\{1,2,3\}$ due to Eq. \eqref{MM_dual}, the desiring result is obtained. 
	
	(2) By the definition of the Walsh transform, for any $(b_1,b_2)\in \gf_{2^t}\times \V$, we have 
	\begin{eqnarray*}
		W_g(b_1,b_2) &=& \sum_{(x,y)\in\gf_{2^t}^2\backslash\{(0,w_1),(0,w_2)\}} (-1)^{h_{\lambda_3}(x,y)+\tr(b_1x+b_2y)} - \left((-1)^{\tr(b_2w_1)} + (-1)^{\tr(b_2w_2)} \right) \\
		&=& \sum_{(x,y)\in\gf_{2^t}^2}    (-1)^{h_{\lambda_3}(x,y)+\tr(b_1x+b_2y)} - 2\left((-1)^{\tr(b_2w_1)} + (-1)^{\tr(b_2w_2)} \right) \\
		&=& (-1)^{\A_3} 2^t - 4,
	\end{eqnarray*}
	where the lase equality holds due to Eqs. \eqref{that f}, \eqref{MM_dual} and $b_2\in\V$. 
	
	(3) Similar to the proof of the above statement, for any $(b_1,b_2)\in \gf_{2^t}\times \V$, we have 
	\begin{eqnarray*}
		W_{f+g}(b_1,b_2) &=& 
		\sum_{(x,y)\in\gf_{2^t}^2}(-1)^{h_{\lambda_1}(x,y)h_{\lambda_2}(x,y)+h_{\lambda_3}(x,y)+\tr(b_1x+b_2y)} - 4 \\
		&=& \sum_{(x,y)\in\gf_{2^t}^2} (-1)^{h_{\lambda_1}(x,y)h_{\lambda_2}(x,y)+h_{\lambda_1}(x,y)+h_{\lambda_2}(x,y)+\tr(b_1x+b_2y)}-4, 
	\end{eqnarray*}
	where the second equality holds due to $\lambda_3=\lambda_1+\lambda_2$. Moreover, by Lemma Corollary \ref{Wal_f1f2+f1+f2}, the desiring result can be obtained.
\end{proof}

In the following, we first prove that the Boolean functions $f$ and $g$ in Definition \ref{fgF} satisfy the conditions in Definition \ref{fgvc}.
\begin{Lemma}\label{s-condition}
	Let $f$ and $g$ be defined as in Definition \ref{fgF}. Then $f$ and $g$ satisfy the conditions in Definition \ref{fgvc}, where $G=g$, $m=2t,k=2$.
\end{Lemma}

\begin{proof}
On one hand, note that $f=h_{\lambda_1}h_{\lambda_2}$, where $h_{\lambda_1} =\tr(\lambda_1 x\pi(y))$ and $h_{\lambda_2} =\tr(\lambda_2 x\pi(y))$ belongs to the MM class. By Eq. \eqref{MM_dual}, we have that $h_{\lambda_1} = \tr(\lambda_1^{-1}y\pi^{-1}(x))$ and $h_{\lambda_2}^{*} = \tr(\lambda_2^{-1}y\pi^{-1}(x))$. It follows from Proposition \ref{bvg-t} (2) that the Boolean function $f$ defined in Eq. (\ref{tf(x)1}) satisfies the conditions $(\mathbf{c}_1)$ and $(\mathbf{c}_2)$ in Definition \ref{fgvc}.
	
On the other hand, the function $g$ is actually
$$g(x,y) = \begin{cases}
	\tr(\lambda_3 x \pi(y)), & ~\text{if $(x,y)\notin\{ (0,w_1), (0,w_2) \}$} \\
	\tr(\lambda_3 x \pi(y)) + 1, & ~\text{if  $(x,y)\in\{ (0,w_1), (0,w_2) \}$.}
\end{cases}$$	
Thus $g(0,0)=0$ and it is clear that $vf+ug$ is not affine for any nonzero $(v,u)\in\gf_2\times\gf_2$. From Lemma \ref{Wal_f_g_f+g} (2) and Lemma \ref{f_minimal}, it follows that $\mathcal{C}_g$ is a minimal code.
\end{proof}

 
We now compute the partial Walsh transforms of the vectorial Boolean function $F(x,y)$ defined as in Definition \ref{fgF}.

\begin{Lemma}\label{tt_1}
	Let $f(x,y)$, $g(x,y)$ and $F(x,y)$ be defined in Definition \ref{fgF}. Let $\mathcal{V}$ be defined in Lemma \ref{Wal_f_g_f+g}.  Then for any $(a_1,a_2,b_1, b_2)\in\mathbb{F}_{2}\times\mathbb{F}_{2}\times\mathbb{F}_{2^{t}}\times\mathcal{V}$ and $\mathrm{i}\in\{0,1\}$, the values of $W_{F}(a_1,a_2,b_1, b_2)$ are given in Tables \ref{tsdtabel1} and \ref{tsdtabel2} for $(\lambda_1+\lambda_2)^{-1}=\lambda_1^{-1}+\lambda_2^{-1}$ or not, respectively, where $(\mathbb{F}_{2^{t}}\times\mathcal{V})^* = \mathbb{F}_{2^{t}}\times\mathcal{V}\backslash\{(0,0)\}$,  $\mathcal{A}_1,\mathcal{A}_2$ and $\mathcal{A}_3$ are defined as in Lemma \ref{Wal_f_g_f+g}, and $\wt(\A_1,\A_2,\A_3)$ means the number of $i$ such that $\A_i=1$.
\end{Lemma}
\begin{table}[t]
	\begin{center}
		\caption{~The values of $W_{F}(a_1,a_2,b_1, b_2)$ when  $(\lambda_1+\lambda_2)^{-1}=\lambda_1^{-1}+\lambda_2^{-1}$ }\label{tsdtabel1}
		\begin{tabular}{cc}
			\hline
			\makecell[c]{Value}  &  \makecell[c]{Condition}   \\
			\hline
			$2^m$  &$a_1=a_2=b_1=b_2=0$  \\
			\hline
			$0$  &  $a_1=a_2=0,(b_1,b_2)\in(\mathbb{F}_{2^{t}}\times\mathcal{V})^*$      \\
			\hline
			$(-1)^{\mathrm{i}}2^t-4$  &  $a_1=0, a_2=1, (b_1,b_2)\in\mathbb{F}_{2^{t}}\times\mathcal{V}$, $\mathcal{A}_3=\mathrm{i}$      \\
			\hline
			$3\cdot2^{t-1}-2^{m-1}-4$  & $a_1=a_2=1, b_1=b_2=0$     \\
			\hline
			$3\cdot2^{t-1}-4$  &  $a_1=a_2=1, (b_1,b_2)\in(\mathbb{F}_{2^{t}}\times\mathcal{V})^*,(\mathcal{A}_1, \mathcal{A}_2)=(0,0)$  \\
			\hline
			$-2^{t-1}-4$  &  $a_1=a_2=1, (b_1,b_2)\in(\mathbb{F}_{2^{t}}\times\mathcal{V})^*,(\mathcal{A}_1, \mathcal{A}_2)\in\{(\mathrm{i}+1,\mathrm{i}),(1,1)\}$  \\
			\hline
			$2^{m-1}+2^{t-1}$  &  $a_1=1, a_2=0, (b_1,b_2)=(0,0)$  \\
			\hline
			$-3\cdot2^{t-1}$  &  $a_1=1, a_2=0, (b_1,b_2)\in(\mathbb{F}_{2^{t}}\times\mathcal{V})^*,(\mathcal{A}_1, \mathcal{A}_2)=(1,1)$  \\
			\hline
			$2^{t-1}$  &  $a_1=1, a_2=0, (b_1,b_2)\in(\mathbb{F}_{2^{t}}\times\mathcal{V})^*,(\mathcal{A}_1, \mathcal{A}_2)\in\{(\mathrm{i}+1,\mathrm{i}),(0,0)\}$  \\
			\hline
		\end{tabular}
	\end{center}
\end{table}

\begin{table}[t]
	\begin{center}
		\caption{~The values of $W_{F}(v,u,b_1, b_2)$ when  $(\lambda_1+\lambda_2)^{-1}\ne\lambda_1^{-1}+\lambda_2^{-1}$ }\label{tsdtabel2}
		\begin{tabular}{cc}
			\hline
			\makecell[c]{Value}  &  \makecell[c]{Condition}   \\
			\hline
			$2^m$  &$a_1=a_2=b_1=b_2=0$  \\
			\hline
			$0$  &  $a_1=a_2=0,(b_1,b_2)\in(\mathbb{F}_{2^{t}}\times\mathcal{V})^*$      \\
			\hline
			$(-1)^{\mathrm{i}}2^t-4$  &  $a_1=0, a_2=1, (b_1,b_2)\in\mathbb{F}_{2^{t}}\times\mathcal{V}$, $\mathcal{A}_3=\mathrm{i}$      \\
			\hline
			$3\cdot2^{t-1}-2^{m-1}-4$  & $a_1=a_2=1, b_1=b_2=0$     \\
			\hline
			$(-1)^{\mathrm{i}}3\cdot2^{t-1}-4$  &$a_1=a_2=1, (b_1,b_2)\in(\mathbb{F}_{2^{t}}\times\mathcal{V})^*, \wt(\mathcal{A}_1, \mathcal{A}_2, \mathcal{A}_3)=3\mathrm{i}$  \\
			\hline
			$(-1)^{\mathrm{i}}2^{t-1}-4$  &  $a_1=a_2=1, (b_1,b_2)\in(\mathbb{F}_{2^{t}}\times\mathcal{V})^*, \wt(\mathcal{A}_1, \mathcal{A}_2, \mathcal{A}_3)=1+\mathrm{i}$      \\
			\hline
			$2^{m-1}+2^{t-1}$  &  $a_1=1, a_2=0, (b_1,b_2)=(0,0)$      \\
			\hline
			$(-1)^{\mathrm{i}}2^{t-1}$  & $a_1=1, a_2=0, (b_1,b_2)\in(\mathbb{F}_{2^{t}}\times\mathcal{V})^*,(\mathcal{A}_1, \mathcal{A}_2, \mathcal{A}_3) \in \{(\mathrm{i},\mathrm{i},\mathrm{i}),(1+\mathrm{i},\mathrm{i},1+\mathrm{i}) ,(\mathrm{i},1+\mathrm{i},1+\mathrm{i})\}$     \\
			\hline
			$(-1)^{\mathrm{i}}3\cdot2^{t-1}$  &  $a_1=1, a_2=0, (b_1,b_2)\in(\mathbb{F}_{2^{t}}\times\mathcal{V})^*,(\mathcal{A}_1, \mathcal{A}_2, \mathcal{A}_3)=(\mathrm{i},\mathrm{i},1+\mathrm{i})$  \\
			\hline
		\end{tabular}
	\end{center}
\end{table}

\begin{proof}
	According to the definition of the Walsh transform, for any $(a_1,a_2,b_1, b_2)\in\mathbb{F}_{2}\times\mathbb{F}_{2}\times\mathbb{F}_{2^{t}}\times\mathcal{V}$, we have 
	\begin{eqnarray}\label{Walsh_F}
		W_F(a_1,a_2,b_1,b_2) &=& \sum_{(x,y)\in\gf_{2^t}^2} (-1)^{a_1f(x,y)+a_2g(x,y)+\tr(b_1x+b_2y)} \notag \\
		&=& \begin{cases}
			\sum_{(x,y)\in\gf_{2^t}^2} (-1)^{\tr(b_1x+b_2y)}, &~\text{if}~a_1=0,a_2=0,  \\
			W_f(b_1,b_2), &~\text{if}~a_1=1,a_2=0, \\
			W_g(b_1,b_2), &~\text{if}~a_1=0,a_2=1, \\
			W_{f+g}(b_1,b_2), &~\text{if}~a_1=1,a_2=1.\\
		\end{cases}
	\end{eqnarray}
	
\textbf{Case I:}	If $\lambda_3^{-1}=\lambda_1^{-1}+\lambda_2^{-1}$, i.e., $\A_3=\A_1+\A_2$, then by Lemma \ref{Wal_f_g_f+g}, we have 
	$$W_f(b_1,b_2) =
		\begin{cases}
			2^{m-1}+2^{t-1}, &~\text{if}~ (b_1,b_2)=(0,0), \\
		-3\cdot	2^{t-1}, &~\text{if}~ (b_1,b_2)\neq(0,0) ~\text{and}~(\A_1,\A_2)=(1,1), \\
		-2^{t-1}, &~\text{if}~ (b_1,b_2)\neq(0,0) ~\text{and}~(\A_1,\A_2)\in \{ (0,0), (1,0),(0,1)\}
		\end{cases}
$$
and 
$$ W_{f+g}(b_1,b_2) =  \begin{cases}
	-2^{m-1}+3\cdot 2^{t-1}-4, &~\text{if}~ (b_1,b_2)=(0,0), \\
	3\cdot 2^{t-1} -4, &~\text{if}~ (b_1,b_2)\neq(0,0) ~\text{and}~(\A_1,\A_2)=(0,0),\\
	-\cdot 2^{t-1} -4, &~\text{if}~ (b_1,b_2)\neq(0,0) ~\text{and}~(\A_1,\A_2)\in \{ (1,1), (1,0),(0,1)\}.\\
\end{cases}$$
Then with Eq. \eqref{Walsh_F}, we can obtain all cases in Table \ref{tsdtabel1}. 

\textbf{Case II:} If $\lambda_3^{-1}\neq\lambda_1^{-1}+\lambda_2^{-1}$, i.e., $\A_3 = \A_1+\A_2$ does not always hold for any $(b_1,b_2)\in (\mathcal{V}\times\mathbb{F}_{2^{t}})^*$, all cases in Table \ref{tsdtabel2} can be obtained by Eq. \eqref{Walsh_F} and Lemma \ref{Wal_f_g_f+g}. 
\end{proof}

\begin{Rem}
	In Lemma \ref{tt_1}, $(\lambda_1+\lambda_2)^{-1}=\lambda_1^{-1}+\lambda_2^{-1}$ if and only if $\lambda_1^2+\lambda_2^2=\lambda_1\lambda_2$, which holds if and only if $t$ is even and  $\frac{\lambda_1}{\lambda_2}\in\gf_4\backslash\gf_2$. 
\end{Rem}

	\begin{table}[H]
	\begin{center}
		\caption{The weight distribution of $\mathcal{C}_{F}$ in Theorem \ref{tcode-rmovet1}(1)}\label{treven-t1}
		\begin{tabular}{cc}
			\hline\noalign{\smallskip}
			Weight  &  Multiplicity   \\
			\noalign{\smallskip}
			\hline\noalign{\smallskip}
			$0$  &  1 \\
			$ 3\cdot2^{m-2}-3\cdot2^{t-2}+2$    &  $  1$   \\
			$ 2^{m-1}$  &  $2^{m-2}-1$     \\
			$2^{m-2}-2^{t-2}$  &  $1$    \\
			$2^{m-1}-(-1)^{\mathrm{i}}2^{t-1}+2$ & $2^{m-3}+(-1)^{\mathrm{i}}2^{t-1}$     \\
			$2^{m-1}-3\cdot2^{t-2}+2$  &  $2^{m-4}+3\cdot2^{t-2}-1$    \\
			$2^{m-1}+2^{t-2}+2$  &  $3\cdot2^{m-4}-3\cdot2^{t-2}$    \\
			$2^{m-1}+3\cdot2^{t-2}$  &  $2^{m-4}-2^{t-2}$    \\
			$2^{m-1}-2^{t-2}$  &  $3\cdot2^{m-4}+2^{t-2}-1$    \\
			\noalign{\smallskip}
			\hline
		\end{tabular}
	\end{center}
\end{table}

 \begin{table}[H]
	\begin{center}
		\caption{The weight distribution of $\mathcal{C}_{F}$ in Theorem \ref{tcode-rmovet1}(2)}\label{treven-t2}
		\begin{tabular}{cc}
			\hline\noalign{\smallskip}
			Weight  &  Multiplicity   \\
			\noalign{\smallskip}
			\hline\noalign{\smallskip}
			$0$  &  1 \\
			$ 3\cdot2^{m-2}-3\cdot2^{t-2}+2$    &  $  1$   \\
			$ 2^{m-1}$    &  $ 2^{m-2}-1$   \\
			$ 2^{m-1}-(-1)^{\mathrm{i}}2^{t-1}+2$  &  $2^{m-3}+(-1)^{\mathrm{i}}2^{t-1}$     \\
			$2^{m-1}-(-1)^{\mathrm{i}}3\cdot2^{t-2}+2$  &  $2^{m-5}+(7-8\mathrm{i})2^{t-3}-1+\mathrm{i}$    \\
			$2^{m-1}-(-1)^{\mathrm{i}}2^{t-2}+2$  &  $3\cdot(2^{m-5}-2^{t-3})$    \\
			$2^{m-2}-2^{t-2}$  &  $1$    \\
			$2^{m-1}-(-1)^{\mathrm{i}}2^{t-2}$  &  $3\cdot2^{m-5}+(5-8\mathrm{i})2^{t-3}-1+\mathrm{i}$    \\
			$2^{m-1}-(-1)^{\mathrm{i}}3\cdot2^{t-2}$  &  $2^{m-5}-2^{t-3}$    \\
			\noalign{\smallskip}
			\hline
		\end{tabular}
	\end{center}
\end{table}

Applying Lemmas \ref{s-condition} and \ref{tt_1} to Lemma \ref{tccs}, we can deduce the following result.
\begin{Th}\label{tcode-rmovet1}
	Let $m=2t\ge 6$, notation and conditions be the same as in Lemma \ref{tt_1}. Let $\mathcal{C}_{F}$  be the linear code derived by
    $$ \C_F = \left\{ (  (a_1,a_2) \cdot F(x,y) + b_1x+b_2y )_{(x,y)\in\gf_2^t\times\gf_2^t\backslash\{ (\mathbf{0},\mathbf{0}) \} }: (a_1,a_2)\in\gf_2\times\gf_2, (b_1,b_2)\in \gf_{2^t}\times \V   \right\}, $$
    where $F(x,y)$ is defined in Definition \ref{fgF}. Then $\mathcal{C}_{F}$ is a $[2^{m}-1,m,2^{m-2}-2^{t-2}]$-linear code. Furthermore, the following two statements hold.
    \begin{enumerate}[(1)]
        \item If $t$ is even and $(\lambda_1+\lambda_2)^{-1}=\lambda_1^{-1}+\lambda_2^{-1}$, then $\mathcal{C}_{F}$ is a SO singly-even minimal
	code violating the AB condition, and the weight distribution is given by Table \ref{treven-t1}.
    \item If $t$ is odd, or $t$ is even and $(\lambda_1+\lambda_2)^{-1}\ne\lambda_1^{-1}+\lambda_2^{-1}$, then $\mathcal{C}_{F}$ is a SO singly-even minimal code violating the AB condition, and the weight distribution is given by Table \ref{treven-t2}.
    \end{enumerate}
\end{Th}

\begin{proof}
	Plugging the values of $W_{F}$ from Tables \ref{tsdtabel1} and \ref{tsdtabel2} back into Lemma \ref{li_method}, we can obtain the weights of $\mathcal{C}_{F}$ in Tables \ref{treven-t1} and \ref{treven-t2}, respectively. As for an example, we compute the number, denoted by $N$, of $(b_1,b_2)\in (\mathbb{F}_{2^{t}}\times\mathcal{V})^*$ satisfying $W_{F}(1,1,b_1, b_2)=3\cdot2^{t-1}-4$ in $(\mathbb{F}_{2^{t}}\times\mathcal{V})^*$, i.e, $(\A_1,\A_2)=(0,0)$ from Table \ref{tsdtabel1}. By the basic knowledge of the exponential sum, we have 
\begin{eqnarray*}\label{foc}
N
&=& \frac{1}{4} \sum_{b_1\in\gf_{2^t}} \sum_{b_2\in\V} \sum_{s_1\in\gf_2} (-1)^{\A_1s_1}\sum_{s_2\in\gf_2} (-1)^{\A_2s_2} -1 \\
&=& \frac{1}{4} \sum_{b_1\in\gf_{2^t}} \sum_{b_2\in\V} \left(1+(-1)^{\A_1}\right)\left(1+(-1)^{\A_2}\right) -1  \\
&=&\frac{1}{4}\sum_{b_1\in\gf_{2^t}} \sum_{b_2\in\V} \left(1+(-1)^{\A_1} +(-1)^{\A_2} +(-1)^{\A_3}\right) -1 \\
&=&2^{m-4}+3\cdot2^{t-2}-1,
\end{eqnarray*}
where the last equality holds due to the fact that for all $i\in\{1,2,3\},$
$$\sum_{b_1\in\gf_{2^t}} \sum_{b_2\in\V}(-1)^{\A_i} = \sum_{b_2\in\V} \sum_{b_1\in\gf_{2^t}}(-1)^{\tr(b_2\pi^{-1}(\lambda_i^{-1}b_1))} = 2^t.$$
The frequency of the other cases in Table \ref{tsdtabel1} and all cases in Table \ref{treven-t2} can be obtained by the method analogous to that used above.

By applying Tables \ref{tsdtabel1} and \ref{tsdtabel2} to Lemmas \ref{CF_SO} and \ref{tccs}, we can deduce that the codes $\mathcal{C}_{F}$ are SO singly-even, minimal, and violate the AB condition.
\end{proof}

\begin{Rem}\label{ttremark2}
	 We can easily see that the minimum Hamming distance $d^{\bot}$ of the dual code $\mathcal{C}_{F}^{\bot}$ of the code $\mathcal{C}_{F}$ is 3.
	Furthermore, it follows from \cite{Grassl:codetables} that the infinite family of dual codes $\mathcal{C}^{\bot}_{F}$ of the $\mathcal{C}_{F}$ in Theorem \ref{tcode-rmovet1} includes some codes are optimal, e.g., $[15, 11, 3]$, $[63, 57, 3]$ and $[255, 247, 3]$. 
\end{Rem}
\begin{example}\label{lcdexamples2}
$(1)$ Let $m=2t=8$, define $F(x,y)=(\mathrm{Tr}(\xi xy^2)\mathrm{Tr}(\xi^6 xy^2), \mathrm{Tr}((\xi+\xi^6)x\pi(y))+((y+\xi)^{2^t-1}+(y+\xi^2)^{2^t-1})(x^{2^t-1}+1))$, where $\xi$ is a primitive element in $\mathbb{F}_{2^{4}}$ such that $\xi^{4}+\xi+1=0$. Then by MAGMA, we have $(\xi+\xi^6)^{-1}=\xi^{-1}+\xi^{-6}$. Let $\mathcal{V}=\{x\in\mathbb{F}_{2^t}:\mathrm{Tr}(\xi x)=\mathrm{Tr}(\xi^2 x)=0\}$. Then we can verify by MAGMA that the linear code $\mathcal{C}_{F}$ is a binary SO singly-even minimal $[255, 8, 60]$-linear code violating the AB condition with the weight enumerator is $1+z^{60}+27z^{118}+40z^{122}+51z^{124}+63z^{128}+36z^{134}+24z^{138}+12z^{140}+z^{182}$, which is consistent with Theorem \ref{tcode-rmovet1}(1).

$(2)$ Let $m=2t=10$, and define $F(x,y)=(\mathrm{Tr}(\xi xy^2)\mathrm{Tr}(\xi^2 xy^2), \mathrm{Tr}((\xi+\xi^2)x\pi(y))+((y+\xi)^{2^t-1}+(y+\xi^2)^{2^t-1})(x^{2^t-1}+1))$, where $\xi$ is a primitive element in $\mathbb{F}_{2^{5}}$ such that $\xi^{5}+\xi^{2}+1=0$. Let $\mathcal{V}=\{x\in\mathbb{F}_{2^t}:\mathrm{Tr}(\xi x)=\mathrm{Tr}(\xi^2 x)=0\}$. Then, we can verify by MAGMA that the linear code $\mathcal{C}_{F}$ is a binary SO singly-even minimal $[1023, 10, 248]$-linear code violating the AB condition with the weight enumerator is $1+z^{248}+28z^{488}+59z^{490}+144z^{498}+115z^{504}+84z^{506}+255z^{512}+84z^{520}+84z^{522}+112z^{530}+28z^{536}+28z^{538}+z^{746}$, which is consistent with Theorem \ref{tcode-rmovet1}(2).

\end{example}
}
\section{Constructions of binary SO singly-even minimal linear codes from known binary doubly-even linear codes}
\label{first_general_approach}

In this section, we propose a general approach to constructing binary SO singly-even linear codes from some known binary doubly-even linear codes. The main idea is motivated by \cite[Theorem 2.1]{chen2025minimal}, where Chen et al. constructed minimal linear codes violating the AB condition.

\begin{Lemma}
	\label{general_approach_1}
	Let $\C$ be a binary SO $[n,k,d]$-linear code, $\C^{'}$ be a binary SO $[n^{'},k^{'},d^{'}]$-linear code. Then we can construct $2^{kk^{'}}$ different binary SO $[n+n^{'},k]$-linear codes.
\end{Lemma}

\begin{proof}
	Let the generator matrix of $\C$ be
	$$G = \begin{pmatrix}
		\bu_1 \\
		\bu_2 \\
		\vdots \\
		\bu_k
	\end{pmatrix}.$$ For any $\bv_1,\bv_2,\ldots,\bv_k\in\C^{'}$, let
	$$\bar{G} = \begin{pmatrix}
		\bu_1 ~|~ \bv_1 \\
		\bu_2 ~|~ \bv_2 \\
		\vdots \\
		\bu_k ~|~ \bv_k
	\end{pmatrix}$$ and $\bar{\C}$ is the linear code generated by $\bar{G}$. The length and dimension of $\bar{\C}$ are trivially $n+n^{'}$ and $k$, respectively. Furthermore, by Proposition \ref{simple_lemma}, it is easy to get that $\bar{\C}$ is SO since $\C$ and $\C^{'}$ are both SO.
	
	Finally, since $\#\C^{'}=2^{k^{'}}$, the possibilities of $(\bv_1,\bv_2,\ldots,\bv_k)$ are $(2^{k^{'}})^k=2^{kk^{'}}$. Moreover, it is trivial that we can get different binary SO linear codes from different $(\bv_1,\bv_2,\ldots,\bv_k)$, and thus we can construct $2^{kk^{'}}$ different binary SO $[n+n^{'},k]$-linear codes.
\end{proof}

\begin{Rem}
	\label{general_approach_1_minimal}
	In the above proof, it is clear that if $\C$ is a minimal linear code, then $\bar{\C}$ is also minimal since $\C$ is a punctured code of $\bar{\C}$.
\end{Rem}

Together with Lemma \ref{general_approach_1} and Proposition \ref{simple_lemma}, we may construct many binary SO singly-even linear codes with $t^{'}$ weights from a binary SO linear code with $t$ weights, where $t^{'}>t$.

\subsection{Constructions of binary SO singly-even linear codes with two weights}

Recall that a binary linear code of length $n=2^m-1$, with parity-check matrix $H$ whose columns consist of all the nonzero vectors of $\gf_2^m$, is called a binary Hamming code of length $2^m-1$, see e.g. \cite{ling2004coding}. Moreover, the dual code of the binary Hamming code is called the binary simplex code, which has one weight. It is well known that the parameter of the binary simplex code is $[2^m-1, m, 2^{m-1}]$. Thus, the binary simplex code is doubly-even, followed by SO. In the following, we applied the binary simplex code to Lemma \ref{general_approach_1} to construct SO singly-even linear codes with two weights.  Moreover, in some cases, our codes are minimal and violate the AB condition.

\begin{Th}
	\label{two_weight}
	Let $\C$ be the binary simplex $[2^m-1,m,2^{m-1}]$-linear code with $m\ge 2$ and a generator matrix of $\C$ be $$G = \begin{pmatrix}
		\bu_1 \\
		\bu_2 \\
		\vdots \\
		\bu_m
	\end{pmatrix}.$$ Let $n^{'}$ be a positive integer with $n^{'}\equiv 2\pmod 4$,
	$$\bar{G} = \begin{pmatrix}
		\bu_1 ~|~ \mathbf{1}_{n^{'}} \\
		\bu_2 ~|~ \mathbf{0}_{n^{'}} \\
		\vdots \\
		\bu_m ~|~ \mathbf{0}_{n^{'}}
	\end{pmatrix},$$ and
	$\bar{\C}$ be the linear code generated by $\bar{G}$. Then
	$\bar{\C}$ is a binary SO singly-even $[2^m-1+n^{'}, m, 2^{m-1}]$-linear code with the weight distribution in Table \ref{WD_two_weight}.
	
	\begin{table}[h]
		\centering
		\caption{The weight distribution of the codes $\bar{\C}$ in Theorem \ref{two_weight}}	\label{WD_two_weight}
		\begin{tabular}{cc}
			\hline
			Weight & Multiplicity \\
			\hline
			$0$ & $1$\\
			$2^{m-1}$ & $2^{m-1}-1$  \\
			$2^{m-1}+n^{'}$ & $2^{m-1}$ \\
			\hline
		\end{tabular}
	\end{table}
	
	Moreover, if $n^{'}>2^{m-1}$, $\bar{\C}$ is a minimal linear code violating the AB condition.
\end{Th}

\begin{proof}
	Let $\C^{'}=\{\mathbf{1}_{n^{'}}, \mathbf{0}_{n^{'}}  \}$. Then it is trivial that $\C^{'}$ is a binary SO $[n^{'},1,n^{'}]$-linear code. Thus by Lemma \ref{general_approach_1}, $\bar{\C}$ is a binary $[2^m-1+n^{'}, m]$-linear code. As for the weight distribution, for one thing, from the matrix $\bar{G}$, we can see that all nonzero codewords in $\bar{\C}$ generated by $\{(\bu_2 | \mathbf{0}_{n^{'}}), (\bu_3 | \mathbf{0}_{n^{'}}), \ldots,  (\bu_m | \mathbf{0}_{n^{'}})\}$ have the weight $2^{m-1}$; for the other thing, the rest codewords, i.e., having the form $(\bu_1 | \mathbf{1}_{n^{'}}) + a_2 (\bu_2 | \mathbf{0}_{n^{'}}) + a_3 (\bu_3 | \mathbf{0}_{n^{'}}) + \cdots + a_m (\bu_m | \mathbf{0}_{n^{'}})$ with $a_2,a_3,\ldots, a_m\in\gf_2$, have the weight $2^{m-1}+n^{'}$. Thus, we get the desired weight distribution of $\bar{\C}$ in Table \ref{WD_two_weight}.
	
	Finally, in $\bar{\C}$, $w_{\min} = 2^{m-1}, w_{\max} = 2^{m-1}+n^{'}$.  By Lemma \ref{minimal}, $\bar{\C}$ is minimal since $w_{\max} \neq 2 w_{\min}$ if $n^{'}>2^{m-1}$. In addition, when $n^{'}>2^{m-1}$, $\frac{w_{\min}}{w_{\max}}<\frac{1}{2}$, which means that $\bar{\C}$ is a minimal linear code violating the AB condition.
\end{proof}

The following numerical results are consistent with Theorem \ref{two_weight}.

\begin{example}
	\label{ex_two_weight_1}
	In Theorem \ref{two_weight}, let $m=3$ and
	$$G = \begin{pmatrix}
		0 & 0 & 0 & 1 & 1 & 1 & 1 \\
		0 & 1 & 1 & 0 & 0 & 1 & 1 \\
		1 & 0 & 1 & 0 & 1 & 0 & 1 \\
	\end{pmatrix}.$$
	Then $G$ is a generator matrix of $\C$. Let $n^{'}=6$ and
	$$\bar{G} = \begin{pmatrix}
		0 & 0 & 0 & 1 & 1 & 1 & 1 & 1 & 1 & 1 & 1 & 1 & 1 \\
		0 & 1 & 1 & 0 & 0 & 1 & 1 & 0 & 0 & 0 & 0 & 0 & 0 \\
		1 & 0 & 1 & 0 & 1 & 0 & 1 & 0 & 0 & 0 & 0 & 0 & 0 \\
	\end{pmatrix}.$$
	Then by MAGMA, we get that $\bar{\C}$ generated by $\bar{G}$ is a binary SO singly-even $[13, 3, 4]$-linear code with the weight enumerator $1+3z^4+4z^{10}$.
\end{example}

\begin{example}
	\label{ex_two_weight_2}
	In Theorem \ref{two_weight}, let $m=4$ and
	$$G = \begin{pmatrix}
		1 & 0 & 0 & 0 & 1 & 0 & 0 & 1 & 1 & 0 & 1 & 0 & 1 & 1 & 1  \\
		0 & 1 & 0 & 0 & 1 & 1 & 0 & 1 & 0 & 1 & 1 & 1 & 1 & 0 & 0  \\
		0 & 0 & 1 & 0 & 0 & 1 & 1 & 0 & 1 & 0 & 1 & 1 & 1 & 1 & 0  \\
		0 & 0 & 0 & 1 & 0 & 0 & 1 & 1 & 0 & 1 & 0 & 1 & 1 & 1 & 1  \\
	\end{pmatrix}.$$
	Then $G$ is a generator matrix of $\C$. Let $n^{'} = 10$ and
	$$\bar{G} = \begin{pmatrix}
		1 & 0 & 0 & 0 & 1 & 0 & 0 & 1 & 1 & 0 & 1 & 0 & 1 & 1 & 1 & 1 & 1 & 1 & 1 & 1 & 1 & 1 & 1 & 1 & 1 \\
		0 & 1 & 0 & 0 & 1 & 1 & 0 & 1 & 0 & 1 & 1 & 1 & 1 & 0 & 0 & 0 & 0 & 0 & 0 & 0 & 0 & 0 & 0 & 0 & 0 \\
		0 & 0 & 1 & 0 & 0 & 1 & 1 & 0 & 1 & 0 & 1 & 1 & 1 & 1 & 0 & 0 & 0 & 0 & 0 & 0 & 0 & 0 & 0 & 0 & 0  \\
		0 & 0 & 0 & 1 & 0 & 0 & 1 & 1 & 0 & 1 & 0 & 1 & 1 & 1 & 1 & 0 & 0 & 0 & 0 & 0 & 0 & 0 & 0 & 0 & 0  \\
	\end{pmatrix}.$$
	Then by MAGMA, we get that $\bar{\C}$ generated by $\bar{G}$ is a binary SO singly-even $[25, 4, 8]$-linear code with the weight enumerator $1+7z^8+8z^{18}$.
	Moreover, $\bar{\C}$ is a minimal linear code violating the AB condition.
\end{example}

\subsection{ Constructions of binary SO singly-even linear codes with four weights}
We now construct binary SO singly-even linear codes with four weights from SO doubly-even linear codes with one weight or two weights.

\begin{Th}
	\label{four_weight_1}
	Let $\C$ be the binary simplex $[2^m-1,m,2^{m-1}]$-linear code with $m\ge3$ and a generator matrix of $\C$ be $$G = \begin{pmatrix}
		\bu_1 \\
		\bu_2 \\
		\vdots \\
		\bu_m
	\end{pmatrix}.$$ Let $n^{'},n^{''}$ be even positive integers with $n^{'}\equiv 2\pmod 4$ and $n^{''}<n^{'}$,
	$$\bar{G} = \begin{pmatrix}
		\bu_1 ~|~ \mathbf{1}_{n^{''}} ~|~ \mathbf{1}_{n^{'}-n^{''}}\\
		\bu_2 ~|~ \mathbf{1}_{n^{''}} ~|~ \mathbf{0}_{n^{'}-n^{''}}\\
		\bu_3 ~|~ \mathbf{0}_{n^{''}} ~|~ \mathbf{0}_{n^{'}-n^{''}} \\
		\vdots \\
		\bu_m ~|~ \mathbf{0}_{n^{''}} ~|~ \mathbf{0}_{n^{'}-n^{''}}
	\end{pmatrix},$$ and
	$\bar{\C}$ be the linear code generated by $\bar{G}$. Then
	$\bar{\C}$ is a binary SO singly-even $[2^m-1+n^{'}, m, 2^{m-1}]$-linear code with the weight distribution in Table \ref{WD_four_weight_1}.
	
	\begin{table}[h]
		\centering
		\caption{The weight distribution of the codes $\bar{\C}$ in Theorem \ref{four_weight_1}}	\label{WD_four_weight_1}
		\begin{tabular}{cc}
			\hline
			Weight & Multiplicity \\
			\hline
			$0$ & $1$\\
			$2^{m-1}$ & $2^{m-2}-1$  \\
			$2^{m-1}+n^{''}$ & $2^{m-2}$ \\
			$2^{m-1} + n^{'} - n^{''}$ & $2^{m-2}$ \\
			$2^{m-1}+n^{'}$ & $2^{m-2}$ \\
			\hline
		\end{tabular}
	\end{table}
	
	Moreover, if $n^{'}>2^{m-1}$, then $\bar{\C}$ is a minimal linear code violating the AB condition.
\end{Th}

\begin{proof}
	Let $\C^{'} = \{ \mathbf{1}_{n^{'}}, (\mathbf{1}_{n^{''}} ~|~ \mathbf{0}_{n^{'}-n^{''}}), (\mathbf{0}_{n^{''}} ~|~ \mathbf{1}_{n^{'}-n^{''}}),   \mathbf{0}_{n^{'}} \}$. Then it is easy to check that $\C^{'}$ is a binary SO $[n^{'},2]$-linear code code. By Lemma \ref{general_approach_1}, $\bar{\C}$ is a binary $[2^m-1+n^{'}, m]$-linear code. There are four cases of the weight distribution.
	
	(1) All nonzero codewords in $\bar{\C}$ generated by $\{ (\bu_3| \mathbf{0}_{n^{'}}), \ldots,  (\bu_m| \mathbf{0}_{n^{'}}) \}$ have the weight $2^{m-1}$.
	
	(2) All codewords in $\bar{\C}$ of the form $ (\bu_1| \mathbf{1}_{n^{'}}) + a_3 (\bu_3| \mathbf{0}_{n^{'}}) + \cdots + a_m (\bu_m| \mathbf{0}_{n^{'}})$ have the weight $2^{m-1}+n^{'}$.
	
	(3) All codewords in $\bar{\C}$ of the form $ (\bu_1| \mathbf{1}_{n^{''}} | \mathbf{0}_{n^{'}-n^{''}}) + a_3 (\bu_3| \mathbf{0}_{n^{'}}) + \cdots + a_m (\bu_m| \mathbf{0}_{n^{'}})$ have the weight $2^{m-1}+n^{'}$.
	
	(4) All codewords in $\bar{\C}$ of the form $ (\bu_1| \mathbf{0}_{n^{''}} | \mathbf{1}_{n^{'}-n^{''}}) + a_3 (\bu_3| \mathbf{0}_{n^{'}}) + \cdots + a_m (\bu_m| \mathbf{0}_{n^{'}})$ have the weight $2^{m-1}+n^{'}-n^{''}$.
	
	Thus,  we get the desired weight distribution of $\bar{\C}$ in Table \ref{four_weight_1}.
	
	Finally, $\bar{\C}$ is minimal by the fact that $\C$ is minimal and Remark \ref{general_approach_1_minimal}. Moreover, when $n^{'}>2^{m-1}$, we have $\frac{w_{\min}}{w_{\max}}<\frac{1}{2}$, which means that $\bar{\C}$ is a minimal linear code violating the AB condition.
\end{proof}

The following numerical results are consistent with Theorem \ref{four_weight_1}.

\begin{example}
	In Example \ref{ex_two_weight_1}, let $n^{'}=6$, $n^{''}=2$ and
	$$\bar{G} = \begin{pmatrix}
		0 & 0 & 0 & 1 & 1 & 1 & 1 & 1 & 1 & 1 & 1 & 1 & 1 \\
		0 & 1 & 1 & 0 & 0 & 1 & 1 & 1 & 1 & 0 & 0 & 0 & 0 \\
		1 & 0 & 1 & 0 & 1 & 0 & 1 & 0 & 0 & 0 & 0 & 0 & 0 \\
	\end{pmatrix}.$$
	Then by MAGMA, we get that $\bar{\C}$ generated by $\bar{G}$ is a binary SO singly-even $[13, 3, 4]$-linear code with the weight enumerator $1+z^4+2z^6+2z^8+2z^{10}$.
	Moreover, $\bar{\C}$ is a minimal linear code violating the AB condition.
\end{example}

\begin{example}
	In Example \ref{ex_two_weight_2}, let $n^{'}=10$, $n^{''}=4$ and
	$$\bar{G} = \begin{pmatrix}
		1 & 0 & 0 & 0 & 1 & 0 & 0 & 1 & 1 & 0 & 1 & 0 & 1 & 1 & 1 & 1 & 1 & 1 & 1 & 1 & 1 & 1 & 1 & 1 & 1 \\
		0 & 1 & 0 & 0 & 1 & 1 & 0 & 1 & 0 & 1 & 1 & 1 & 1 & 0 & 0 & 1 & 1 & 1 & 1 & 0 & 0 & 0 & 0 & 0 & 0 \\
		0 & 0 & 1 & 0 & 0 & 1 & 1 & 0 & 1 & 0 & 1 & 1 & 1 & 1 & 0 & 0 & 0 & 0 & 0 & 0 & 0 & 0 & 0 & 0 & 0  \\
		0 & 0 & 0 & 1 & 0 & 0 & 1 & 1 & 0 & 1 & 0 & 1 & 1 & 1 & 1 & 0 & 0 & 0 & 0 & 0 & 0 & 0 & 0 & 0 & 0  \\
	\end{pmatrix}.$$
	Then by MAGMA, we get that $\bar{\C}$ generated by $\bar{G}$ is a binary SO singly-even $[25, 4, 8]$-linear code with the weight enumerator $1+3z^8+4z^{12}+4z^{14}+4z^{18}$.
	Moreover, $\bar{\C}$ is a minimal linear code violating the AB condition.
\end{example}

In the following, we start with an SO binary linear code with two weights. Note that there are many binary linear codes with two weights, see e.g. \cite{wu2019optimal,hyun2020infinite,qi2015binary,chen2024linear}. We here consider those from bent functions.

\begin{Lemma}
	\label{Code_Bent}
	\cite{ding2016construction}
	Let $f$ be a Boolean function from $\gf_{2^m}$ to $\gf_2$ with $f(\mathbf{0})=0$, where $m\ge 4$ and is even. Then $\C_{D_f}$ defined as in  \eqref{second_method} is a binary $\left[ n_f, m, \frac{n_f-2^{(m-2)/2}}{2} \right]$ two-weight linear code with the weight distribution in Table \ref{WD_Bent}, where $n_f=2^{m-1}\pm 2^{(m-2)/2}$, if and only if $f$ is bent.
	\begin{table}[h]
		\centering
		\caption{The weight distribution of the codes $\bar{\C}$ in Lemma \ref{Code_Bent}}	\label{WD_Bent}
		\begin{tabular}{cc}
			\hline
			Weight & Multiplicity \\
			\hline
			$0$ & $1$\\
			$\frac{n_f}{2}-2^{\frac{m-4}{2}}$ & $\frac{2^m-1-n_f2^{-\frac{m-2}{2}}}{2}$  \\
			$\frac{n_f}{2}+2^{\frac{m-4}{2}}$ & $\frac{2^m-1+n_f2^{-\frac{m-2}{2}}}{2}$ \\
			\hline
		\end{tabular}
	\end{table}
\end{Lemma}

\begin{Th}
	\label{four_weight_2}
	Let $f$ be a bent function and $\C_{D_f}$ be the binary linear code defined as in Lemma \ref{Code_Bent}. Let
	$$G_{D_f} = \begin{pmatrix}
		\bu_1 \\
		\bu_2 \\
		\vdots \\
		\bu_m
	\end{pmatrix}$$ be a generator matrix of $\C_{D_f}$.  Let $n^{'}$ be a positive integer with $n^{'}\equiv 2\pmod 4$,
	\begin{equation}
		\label{GDf}
		\bar{G}_{D_f} = \begin{pmatrix}
			\bu_1 ~|~ \mathbf{1}_{n^{'}} \\
			\bu_2 ~|~ \mathbf{0}_{n^{'}} \\
			\vdots \\
			\bu_m ~|~ \mathbf{0}_{n^{'}}
		\end{pmatrix},
	\end{equation}
	and
	$\bar{\C}_{D_f}$ be the linear code generated by $\bar{G}_{D_f}$. Then $\bar{\C}_{D_f}$ is a binary SO singly-even $\left[ n_f+n^{'}, m, \frac{n_f-2^{(m-2)/2}}{2} \right]$ four-weight linear code with the weight distribution in Table \ref{WD_four_weight_2} if $n_f=2^{m-1} + 2^{(m-2)/2}$; or Table \ref{WD_four_weight_3} if $n_f=2^{m-1} - 2^{(m-2)/2}$.
	\begin{table}[h]
		\centering
		\caption{The weight distribution of the codes $\bar{\C}_{D_f}$ in Theorem \ref{four_weight_2} if $n_f=2^{m-1} + 2^{(m-2)/2}$}	\label{WD_four_weight_2}
		\begin{tabular}{cc}
			\hline
			Weight & Multiplicity \\
			\hline
			$0$ & $1$\\
			$2^{m-2}$ & $2^{m-2}-1$  \\
			$2^{m-2} + n^{'}$ & $2^{m-2}-2^{\frac{m-2}{2}}$ \\
			$2^{m-2}+ 2^{\frac{m-2}{2}}$ & $2^{m-2}$ \\
			$2^{m-2}+2^{\frac{m-2}{2}}+n^{'}$ & $2^{m-2}+2^{\frac{m-2}{2}}$ \\
			\hline
		\end{tabular}
	\end{table}
	
	\begin{table}[h]
		\centering
		\caption{The weight distribution of the codes $\bar{\C}_{D_f}$ in Theorem \ref{four_weight_2} if $n_f=2^{m-1} - 2^{(m-2)/2}$}	\label{WD_four_weight_3}
		\begin{tabular}{cc}
			\hline
			Weight & Multiplicity \\
			\hline
			$0$ & $1$\\
			$2^{m-2}- 2^{\frac{m-2}{2}}$ & $2^{m-2}$  \\
			$2^{m-2} - 2^{\frac{m-2}{2}} + n^{'}$ & $2^{m-2} - 2^{\frac{m-2}{2}}$ \\
			$2^{m-2} $ & $2^{m-2}-1$ \\
			$2^{m-2}+n^{'}$ & $2^{m-2}+2^{\frac{m-2}{2}}$ \\
			\hline
		\end{tabular}
	\end{table}
\end{Th}
Moreover, if $n^{'}>2^{m-2} - 2^{\frac{m-2}{2}}$, then  $\bar{\C}_{D_f}$ is a minimal linear code violating the AB condition.
\begin{proof}
	We only give the proof of the weight distribution of $\bar{\C}_{D_f}$ if $n_f=2^{m-1} + 2^{(m-2)/2}$. The other case is similar and thus omitted.
	
	When  $n_f=2^{m-1} + 2^{(m-2)/2}$, by Lemma \ref{Code_Bent},  $\C_{D_f}$ is a binary $\left[2^{m-1} + 2^{(m-2)/2}, m, 2^{m-2}\right]$-linear code with the weight distribution
	$$1+\left(2^{m-1}-2^{\frac{m-2}{2}}-1\right) z^{2^{m-2}} + \left( 2^{m-1}+2^{\frac{m-2}{2}} \right) z^{2^{m-2}+2^{\frac{m-2}{2}}}.$$
	Therefore, the weights of codewords of $\bar{\C}_{D_f}$ belong to
	$$ \left\{ w_1 = 2^{m-2}, w_2 = 2^{m-2} + n^{'}, w_3 =  2^{m-2}+2^{\frac{m-2}{2}}, w_4 = 2^{m-2}+2^{\frac{m-2}{2}} + n^{'}  \right\}. $$
	
	In the following, we determine the multiplicities of all weights. Assume that the multiplicity of $w_1$ is equal to $A_{w_1} = \Delta$. Then by Proposition \ref{simple_lemma} (2), all codewords with weights divisible by $4$ form a linear code with dimension $m-1$. Thus the multiplicity of $w_3$ is equal to $A_{w_3} = 2^{m-1}-1-\Delta$. Moreover, all codewords with the weight $w_2$ are from that with the weight $w_1$ and thus the multiplicity of $w_2$ is equal to $A_{w_2} = 2^{m-1}-2^{\frac{m-2}{2}}-1-\Delta$. Similarly, the multiplicity of $w_4$ is equal to $ A_{w_4} = 2^{m-1}+2^{\frac{m-2}{2}} - \left( 2^{m-1}-1-\Delta \right) = 2^{\frac{m-2}{2}}+1 + \Delta$.
	
	In order to compute the value of $\Delta$, we consider the weights of codewords of the dual $\bar{\C}_{D_f}^{\perp}$ of $\bar{\C}_{D_f}$ and then use the first third power moments (see \cite[P. 260]{huffman2010fundamentals}). Assume that the multiplicity of the weight $1$ (resp. $2$) in $\bar{\C}_{D_f}^{\perp}$ is $A_1^{\perp}$ (resp. $A_2^{\perp}$). Firstly, since all weights of codewords in $\bar{\C}_{D_f}^{\perp}$ is even, $A_1^{\perp}=0$.  Secondly, assume that $G$ is in standard form and $\bu = (u_1,u_2,\ldots, u_{n_f+n^{'}})\in\bar{\C}_{D_f}^{\perp}$ with $\wt(\bu)=2$. Since $\C^{\perp}$ does not have codewords with the weight $2$, there must exist a integer $n_f+1 \le i\le n_f+n^{'}$ such that $u_{i}=1$. If there are two integers   $n_f+1 \le i,j\le n_f+n^{'}$ such that $u_{i}=1$, then the possibility of $\bu$ is $\frac{n^{'}(n^{'}+1)}{2}$. If there is exactly one integer $n_f+1 \le i\le n_f+n^{'}$ such that $u_{i}=1$, then $u_1=1$ since $G$ is in standard form and thus the possibility of $\bu$ is $n^{'}$. Hence, we can get $A_2^{\perp} = \frac{n^{'}(n^{'}+1)}{2} + n^{'}$. Finally, by solving the first third power moments, i.e.,
	$$w_1^2A_{w_1} + w_2^2 A_{w_2} + w_3^2 A_{w_3} + w_4^2 A_{w_4} = 2^{m-2}\left[ (n_f+n^{'})(n_f+n^{'}+1) + 2 A_2^{\perp} \right],$$
	we can get $\Delta = 2^{m-2}-1$ and then the desired weight distribution in Table \ref{WD_four_weight_2} has been obtained.
\end{proof}

The following numerical data is consistent with Theorem \ref{four_weight_2}.

\begin{example}
	In Theorem \ref{four_weight_2}, let $m=6$, $\gamma$ be a primitive element of $\gf_{2^6}$, and $f(x)=\tr_{2^6}(\gamma x^3)$ be a bent function.  Let
	\begin{eqnarray*}
		&&   G_{D_f} = \\
		&& \tiny \begin{pmatrix}
			1 & 0 & 0 & 0 & 0 & 0 & 1 & 1 & 0 & 0 & 1 & 1 & 1 & 0 & 1 & 0 & 1 & 0 & 1 & 0 & 1 & 1 & 0 & 0 & 1 & 0 & 0 & 0 & 0 & 0 & 1 & 1 & 1 & 1 & 0 & 0 \\
			0 &  1 & 0 & 0 & 0 & 0 & 1 & 0 & 0 & 0 & 1 & 0 & 0 & 1 & 0 & 1 & 0 & 1 & 0 & 1 & 1 & 0 & 0 & 0 & 1 & 0 & 0 & 1 & 0 & 1 & 1 & 1 & 1 & 1 & 1 & 0   \\
			0 & 0 & 1 & 0 & 0 & 0 & 1 & 1 & 0 & 1 & 0 & 1 & 0 & 1 & 0 & 1 & 0 & 0 & 0 & 1 & 0 & 0 & 0 & 1 & 0 & 1 & 0 & 0 & 1 & 1 & 1 & 1 & 0 & 1 & 0 & 1  \\
			0 & 0 & 0 & 1 & 0 & 0 & 1 & 0 & 1 & 1 & 1 & 1 & 1 & 1 & 1 & 1 & 0 & 0 & 1 & 1 & 0 & 0 & 1 & 1 & 0 & 1 & 0 & 1 & 0 & 0 & 0 & 0 & 0 & 0 & 0 & 0 \\
			0 & 0 & 0 & 0 & 1 & 0 & 1 & 1 & 1 & 1 & 0 & 0 & 1 & 0 & 0 & 0 & 1 & 0 & 0 & 1 & 0 & 1 & 0 & 1 & 0 & 0 & 1 & 0 & 0 & 0 & 1 & 0 & 1 & 1 & 1 & 1 \\
			0 & 0 & 0 & 0 & 0 & 1 & 1 & 0 & 1 & 0 & 0 & 0 & 0 & 1 & 0 & 0 & 0 & 1 & 1 & 1 & 1 & 1 & 1 & 1 & 0 & 0 & 1 & 1 & 1 & 0 & 1 & 0 & 1 & 0 & 0 & 0 \\
		\end{pmatrix}
	\end{eqnarray*}
	be a generator matrix of $\C_{D_f}$  and $n^{'} = 14$, $\bar{G}_{D_f}$ be defined as in \eqref{GDf}. Then, by MAGMA, we can get that
	$\bar{\C}_{D_f}$ is binary SO singly-even $[50, 6, 16]$-linear code with the weight enumerator $1+15 z^{16} + 16 z^{20} + 12 z^{30} + 20 z^{34}$. Moreover, $\bar{\C}_{D_f}$ is a minimal linear code violating the AB condition.
\end{example}

\section{Conclusion}
\label{conclusion}
{ This paper studied the constructions of binary linear codes that are simultaneously SO, singly-even, minimal, non-AB, and have few weights. Our contributions are as follows:
\begin{enumerate}[(1)]
    \item a simple yet powerful characterization for a binary linear code to be SO and singly-even, and necessary and sufficient  conditions for Boolean and vectorial Boolean functions to generate such codes via a standard construction method;
    \item a general method for constructing Boolean functions satisfying the condition in Lemma \ref{Cf_SO_SE}, and several infinite families of binary SO singly-even minimal linear codes violating the AB condition with few weights;
    \item a generic construction of binary linear codes based on vectorial Boolean functions of the form $F(x)=(f(x),G(x))$ with demonstration through several explicit infinite classes, where $f=f_1f_2$ and $f_1,f_2$ are bent Boolean functions; 
    \item a general approach to constructing binary SO singly-even linear codes from SO linear codes, providing three infinite classes of binary SO singly-even minimal few-weight linear codes violating the AB condition.
\end{enumerate}
}
 Finally, we would like to emphasize that using the methods in this paper, one can construct more binary linear codes that are SO, singly-even, minimal, violating the AB condition, and with few weights simultaneously. Moreover, the work of this paper may be generalized to finite fields with any characteristic.

 \bibliographystyle{plain}
 \bibliography{ref}

\end{document}